\documentclass[11pt]{article}
\usepackage[left=1in,top=1in,right=1in,bottom=1in,head=.1in,nofoot]{geometry}

\usepackage{setspace,url,bm,amsmath} 
\usepackage{float}
\usepackage{caption}
\usepackage{subcaption,siunitx,booktabs}

\usepackage{titlesec} 
\usepackage{rotating}
\titlelabel{\thetitle.\quad} 
\titleformat*{\section}{\bf\large}

\usepackage[titletoc,title]{appendix}
\usepackage{graphicx} 
\usepackage{bbm}
\usepackage{latexsym}
\usepackage{xcolor}
\usepackage{authblk}
\usepackage{amsthm}
\usepackage{amsfonts}
\usepackage{bm}
\usepackage{enumitem}
\usepackage{amsmath}
\usepackage{appendix}

\theoremstyle{definition}
\newtheorem{theorem}{Theorem}

\newtheorem{lemma}{Lemma}

\newtheorem{corollary}{Corollary}
\newtheorem{remark}{Remark}
\newtheorem{condition}{Condition}

\usepackage{natbib} 
\bibpunct{(}{)}{;}{a}{}{,} 
\usepackage{enumitem}
\usepackage{tcolorbox}

\usepackage{etoolbox} 
\apptocmd{\sloppy}{\hbadness 10000\relax}{}{} 

\def\T{\textnormal{T}}
\def\var{\textnormal{Var}}
\def\cov{\textnormal{Cov}}

\def \tr{\text{tr}}
\def \Dtil{\widetilde{\mathcal{D}}}

\def \Dtil{\widetilde{\mathcal{D}}}
\def \Dprime{\mathcal{D}^\prime}
\DeclareMathOperator*{\argmax}{arg\!\max}
\DeclareMathOperator*{\argmin}{arg\!\min}
\DeclareUnicodeCharacter{2212}{-}

\newtheorem*{theorem3.2}{Theorem 3.2 (\cite{FRIEDRICH20151})}
\newtheorem*{Definition}{Definition}

{\endtcolorbox}

\title{Optimal allocation of sample size for randomization-based inference from $2^K$ factorial designs}

\author[*]{Arun Ravichandran}
\author[*]{Nicole E. Pashley}
\author[**]{Brian Libgober}
\author[*]{Tirthankar Dasgupta}

\affil[*]{Department of Statistics, Rutgers University}
\affil[**]{Department of Political Science and Law, Northwestern University}

\begin{document}
	
	\date{}
	
	\doublespacing
	
	\maketitle
	
	\begin{abstract}
		Optimizing the allocation of units into treatment groups can help researchers improve the precision of causal estimators and decrease costs when running factorial experiments. However, existing optimal allocation results typically assume a super-population model and that the outcome data comes from a known family of distributions. Instead, we focus on randomization-based causal inference for the finite-population setting, which does not require model specifications for the data or sampling assumptions. We propose exact theoretical solutions for optimal allocation in $2^K$ factorial experiments under complete randomization with A-, D- and E-optimality criteria. We then extend this work to factorial designs with block randomization. We also derive results for optimal allocations when using cost-based constraints. To connect our theory to practice, we provide convenient integer-constrained programming solutions using a greedy optimization approach to find integer optimal allocation solutions for both complete and block randomization. The proposed methods are demonstrated using two real-life factorial experiments conducted by social scientists.
	\end{abstract}

	\section{Introduction}
	\label{sec:intro}
	
	Randomized $2^K$ factorial experiments are conducted to assess the marginal causal effects of $K$ factors, each with two levels, along with their interactions on a response of interest. The two levels are often denoted as the ``high level'' and ``low level'' of the factor \citep{Fisher1935, Yates1937}. With $K$ factors, there are $2^K$ unique treatment combinations to which units can be assigned. In the twentieth century, factorial designs have mostly been discussed in an industrial setting, whereas in recent times, there has been a lot of interest in their application in the social, behavioral and biomedical sciences and randomization-based inference from such designs \citep[e.g.,][]{BDR2016, Egami_Imai2019}.
	
	Randomization-based inference is a useful methodology for drawing inference on causal effects of treatments in a finite population setting \citep[e.g.,][]{Freedman2006, Freedman2008a}. A major advantage of randomization-based inference is that it applies even if the experimental units are not randomly sampled from a larger population, which is the case in most social science experiments \citep{Abadie2020, Olsen2013}. The theory, methods, and applications of randomization-based inference for two-level factorial experiments with a completely randomized treatment assignment mechanism have been developed and discussed \citep[e.g.,][]{DPR2015, Lu2016}. Further, randomization-based inference from experiments with more general factorial structures and complex assignment mechanisms have been discussed in \cite{MDR2018}. Connections between regression-based and randomization-based causal inference from factorial experiments have been studied by \cite{ZhaoDing2021}.
	
	Despite the growing literature in this area, most of the recent research on randomization-based inference of factorial experiments has been confined to the analysis side. On the design side, the main focus has been on rerandomization \citep{BDR2016, li_rerand, morgan2012rerandomization}, which generalizes the idea of blocking by pre-defining an acceptable criterion for randomization based on covariate balance between treatment groups. 
	There have also been extensions to fractional and incomplete factorial designs \citep{pashley2022causal} and to the use of screening steps \citep{shi2023forward}.
	However, the distribution of the total number of experimental units into the $2^K$ treatment groups has not received much attention. Balanced designs that assign equal number of units to the treatment groups are often the default choice, but it is unclear whether they are the ``best'' design under different conditions. 
	One work that does discuss how to allocate units to optimize precision of factorial estimators from the randomization-based perspective is \cite{blackwell2022batch}.
	That work explores the advantages of Neyman-Allocation \citep{Neyman34, Cochran77} by extending the two-stage adaptive design in \cite{Hahn2011} to multiple treatment designs.
	\cite{dai2023clip} similarly explores Neyman-Allocation but within sequential designs.
	
	To motivate this problem, we consider an education experiment from \cite{Angrist2009}, conducted to assess the causal effects of two different interventions, a student support program (SSP) and a student fellowship program (SFP), on the academic performance of freshmen. This is a $2^2$ factorial experiment in which each unit (freshman) can receive only one of four treatment combinations: control (neither of the two), SSP only, SFP only, and SFSP (both). The units were divided into two blocks based on their sex. Table~\ref{tab:Angrist} shows the allocation of units within each block to the four treatment combinations.
	
	\begin{table}[htbp]
		\caption{Allocation of units to treatment combinations} \label{tab:Angrist}
		\begin{center}
			\begin{tabular}{|c|c|c|c|c|}
				\hline
				Sex     & Control & SFP & SSP & SFSP \\  \hline
				Female  &  574 & 150 & 142 & 82 \\
				Male &   432 & 100 & 108 & 68 \\ \hline
				Total & 1006 & 250 & 250 & 150  \\  \hline
			\end{tabular}
		\end{center}
	\end{table}
	
	Clearly, the design is unbalanced, with the highest number of units assigned to control and the fewest to the treatment SFSP. Such an allocation, among other reasons, could be motivated by the budget for experimentation. The question we investigate in this paper is the following: Under what assumptions, conditions, and requirements will such an allocation be the best possible one (in terms of being able to precisely answer scientific questions of interest)?
	
	The problem of finding optimal designs in the context of model-based inference has been extensively studied in the twentieth century \citep[see][for example]{Atkinson2007}. In such settings, optimal designs depend on a postulated outcome model, that may be linear or non-linear. For example, for binary responses, optimal designs based on logistic models are likely to differ from those based on probit models, and depend on unknown model parameters \citep{Mandal2012, Mandal2015, Mandal2016}. We aim to develop optimal designs that are tied to \emph{model-free, randomization-based analysis} for finite and super populations. In addition to being robust to model assumptions, our approach works for continuous as well as binary outcomes as long the finite-population estimand is well-defined. 
	
	This paper is organized as follows: The next section introduces basic notation and estimands for factorial experiments using the potential outcomes framework. In Section \ref{sec:CRDfactorial} we derive optimal allocations of the $N$ units in a population to different treatment combinations under three commonly used optimality criteria for a completely randomized design (CRD). In addition to theoretical results for exact optimal allocations, we also provide numerical algorithms for obtaining integer solutions. In Section \ref{sec:RBDfactorial}, we extend our results for CRDs to the setting of randomized block designs (RBDs). In Section \ref{sec:costbased_optimal_alloc}, we derive optimality results under cost constraints. Two different applications of the proposed methodology, the motivating education experiment and an audit experiment conducted to assess discrimination, are described in Section \ref{sec:examples}. We conclude with a discussion, including opportunities for future work, in Section \ref{sec:discussion}.
	
	\section{$2^K$ factorial experiments under the potential outcomes framework} \label{sec:notation}
	
	Here, we introduce some key definitions and notation from \cite{DPR2015}. Consider a $2^K$ experiment with $N$ units, in which the levels of each of the $K$ factors are denoted by 0 and 1. Each treatment combination is of the form $\mathbf{z}_j = (z_1 \ldots, z_K)$, where $z_k \in \{0,1\}$ for $k \in 1, \ldots, K$. There are $J = 2^K$ treatment combinations arranged in lexicographic order $1, \ldots, J$, where treatment combination $\mathbf{z}_j$ is such that $j = 2^{K-1} z_1 + 2^{K-2} z_2 + \ldots + z_{K} + 1$. In other words, $(z_1 \ldots z_K)$ is a binary representation of integer $j-1$. Thus, for example, in a $2^2$ experiment, the four treatment combinations 00, 01, 10 and 11 are numbered as $j = 1, 2, 3$ and $4$ respectively, and the $8^{th}$ treatment combination in a $2^4$ experiment is $0111$.
	We will just refer to the treatment combination by its number ($j$) in notation below.
	
	For $i = 1, \ldots, N$, under the Stable Unit Treatment Value Assumption or SUTVA \citep{rubin1980}, the $i^{th}$ unit has $J = 2^K$ potential outcomes $Y_i(1), \ldots, Y_i(J)$ corresponding to the $J$ treatment combinations $\mathbf{z}_1, \ldots, \mathbf{z}_J$. Let $\mathbf{Y}_i$ denote the $J \times 1$ vector of potential outcomes for unit $i$.  For unit $i$, the unit-level main effect of factor $k=1, \ldots, K$ is defined as the difference between the averages of potential outcomes for unit $i$ for which the levels of factor $k$ are at levels 0 versus 1. Mathematically, it is a contrast of the form
	$2^{-(K-1)} {\bm \lambda}_k^{\T} \mathbf{Y}_i = 2^{-(K-1)} \sum_{j=1}^J \lambda_{jk} \mathbf{Y}_i(j)$, where $\mathbf{x}^\T$ denotes the transpose of vector $\mathbf{x}$, ${\bm \lambda}_k$ is a $J \times 1$ column-vector with coefficient $\lambda_{jk}$ such that $\lambda_{jk} = -1$ if the level of factor $k$ in $j^{th}$ treatment combination is 0, and $\lambda_{jk} = 1$ otherwise. For all $k=1, \ldots, K$,  ${\bm \lambda}_k$ is a contrast vector, i.e., $\sum_{j=1}^J \lambda_{jk} = 0$.
	
	Proceeding along the lines of \cite{DPR2015}, for unit $i$, we can define ${K \choose 2}$ two-factor interactions, ${K \choose 3}$ three-factor interactions, and finally one $K$-factor interaction as contrasts of the form $2^{-(K-1)} {\bm \lambda}^{\T} \mathbf{Y}_i$, where the contrast vector ${\bm \lambda}$ for any interaction can be derived by element-wise multiplication of the contrast vectors of the corresponding main effects ${\bm \lambda}_k$, for factors involved in the interaction. Denoting the $J-1 = 2^K-1$ contrast vectors for the $J-1$ unit-level factorial effects $\tau_{1i}, \ldots, \tau_{J-1,i}$ by ${\bm \lambda}_1, \ldots, {\bm \lambda}_{J-1}$ respectively, we define the $J \times J$ matrix as
	\begin{equation}
		\bm{L} = ({\bm \lambda}_0, {\bm \lambda}_1, \ldots, {\bm \lambda}_{J-1}), \label{eq:matrixL}
	\end{equation}
	where ${\bm \lambda}_0$ is the $J \times 1$ vector with all elements equal to one. We note that $\mathbf{L}$ is an orthogonal matrix with $\mathbf{L} \mathbf{L}^\T = \mathbf{L}^\T \mathbf{L} =2^{K-1} \mathbf{I}_J$, where $\mathbf{I}_J$ denotes the identity matrix of order $J$. For $i=1, \ldots, N$, let ${\bm \tau}_i = (2\tau_{0i}, \tau_{1i}, \ldots, \tau_{J-1,i})^{\T}$, where $\tau_{0i}$ denotes the average of all potential outcomes for unit $i$. The linear transform between the vector of unit-level potential outcomes $\mathbf{Y}_i$ and the vector of unit-level factorial effects $\bm{\tau}_i$ can be expressed as 
	\begin{equation}
		\bm{\tau}_i = 2^{-(K-1)} \mathbf{L}^{\T} \mathbf{Y}_i. \label{eq:unitleveltau}
	\end{equation} 
	
	Having defined unit-level factorial effects, we now move to their population-level counterparts.
	Let $\overline{\mathbf{Y}} = N^{-1} \sum_{i=1} \mathbf{Y}_i$ and $\bm{\tau} = N^{-1} \sum_{i=1} {\bm \tau}_i$ respectively denote the $J \times 1$ vectors of average potential outcomes and the average factorial effects. Then, averaging (\ref{eq:unitleveltau}) over $i = 1, \ldots, N$, the vector of population-level factorial effects is given by
	\begin{equation}
		\bm{\tau} = 2^{-(K-1)} \mathbf{L}^{\T} \overline{\mathbf{Y}}. \label{eq:popleveltau}
	\end{equation}  
	Note that the first element of $\bm{\tau}$ is twice the average of all potential outcomes.
	
	In a CRD, a pre-assigned numbers of units, $N_{j}$, are randomly assigned to treatment $j$. The experiment generates an $N \times 1$ vector of observed outcomes data from which the vector of factorial effects ${\bm \tau}$ can be unbiasedly estimated. We examine the properties of $\widehat{\bm \tau}$, the unbiased estimator of ${\bm \tau}$, with respect to its randomization distribution, and formulate the problem of optimally allocating the $N$ units to the $J$ treatment combinations in Section~\ref{sec:CRDfactorial}. 
	
	\begin{remark}[A super population perspective]\label{ss:superpop}
		While the finite-population perspective does not depend on any hypothetical data generating process for the outcomes, alternative approaches assume that the potential outcomes are drawn from a, possibly hypothetical, super population. Assuming that $\mathbf{Y}_1, \ldots, \mathbf{Y}_N$ are independent and identically distributed random vectors with $E[\mathbf{Y}_i] = {\bm \mu}$, factorial effects at a super population level are defined as
		\begin{equation}
			\bm{\tau}^{\text{SP}} = 2^{-(K-1)} \mathbf{L}^{\T} {\bm{\mu}}. \nonumber
		\end{equation}
		\cite{DL2017SP} discussed the conceptual and mathematical connections between finite- and super-population inference, showing that while the same estimator commonly used to estimate $\bm{\tau}$ unbiasedly is also an unbiased estimator of $\bm{\tau}^{\text{SP}}$, its sampling variances under the two perspectives are different.
	\end{remark}

	\section{Optimal designs for completely randomized experiments} \label{sec:CRDfactorial}
	
	In a randomized experiment with $N_j$ units assigned to treatment combination $j \in \{1, \ldots, J\}$, only one of the $J$ potential outcomes is observed for unit $i$. This observed outcome is $y_i = Y_i(T_i)$ for $i=1, \ldots, N$, where $T_i$ is the random treatment assignment variable for unit $i$ taking value $j$ if unit $i$ receives treatment $j$. In a CRD, the joint probability distribution of $(T_1, \ldots, T_N)$ is
	\[ 
	P[(T_1, \ldots, T_N) = (t_1, \ldots, t_N)]  = \left\{ \begin{array}{cc} 
		\left( N_1! \ldots N_J!\right) / N! & \mbox{if} \ \sum_{i=1}^N \mathbbm{1}_{\{t_i = j\}} = N_j \ \mbox{for} \ j=1, \ldots, J, \\
		0 & \mbox{otherwise,} 
	\end{array}
	\right.
	\]
	where $\mathbbm{1}_{\{A\}}$ denotes the indicator random variable for set $A$. Let 
	$\overline{y}(j) = N_j^{-1} \sum_{i=1}^N \mathbbm{1}_{\{t_i = j\}} Y_i(j)$ denote the average response for treatment $j$. Let $\overline{\mathbf{y}}$ denote the vector $\left(\overline{y}(1), \overline{y}(2), \ldots,\overline{y}(J)\right)^{\T}$ of observed averages. Substituting $\overline{\mathbf{y}}$ in place of $\overline{\mathbf{Y}}$ in (\ref{eq:popleveltau}), we can unbiasedly estimate the vector of factorial effects as
	\begin{equation}
		\widehat{\bm \tau} = 2^{-(K-1)} \mathbf{L}^{\T} \overline{\mathbf{y}}. \label{eq:est_tau}
	\end{equation}
	
	\cite{Lu2016} derived sampling properties of the estimator $\widehat{\bm \tau}$ with respect to its randomization distribution for the general case of unequal $N_1, \ldots, N_J$. \citeauthor{Lu2016} showed that $\widehat{\bm \tau}$ is an unbiased estimator of  $\bm \tau$, and has the following finite-population covariance matrix:
	\begin{equation}
		\mathbf{V}_{\bm \tau} = \var \left(\widehat{\bm \tau} \right) = \frac{1}{2^{2(K-1)}} \sum_{j=1}^{2^K} \frac{S^2_j}{N_j} \widetilde{\bm  \lambda}_j \widetilde{\bm \lambda_j} ^{\T} - \frac{1}{N(N-1)} \sum_{i=1}^N  \left( \bm{\tau}_i  - \bm{\tau} \right) \left( \bm{\tau}_i  - \bm{\tau} \right)^{\T}, \label{eq:cov_matrix}
	\end{equation}
	where $\widetilde{\bm \lambda}_j$ represents the transpose of row $j$ of the model matrix $\mathbf{L}$ defined in (\ref{eq:matrixL}), ${\bm \tau}_i$ denotes the vector of unit-level factorial effects given by (\ref{eq:unitleveltau}), $\bm{\tau}$ the vector of population-level factorial effects given by (\ref{eq:popleveltau}), and
	\begin{equation}
		S^2_j = \frac{1}{N-1} \sum_{i=1}^N \left(Y_i(j) - \overline{Y}(j) \right)^2 \label{eq:varj}
	\end{equation}
	the variance of all $N$ potential outcomes for treatment $j$ with divisor $N-1$, where $\overline{Y}(j) = N^{-1} \sum_{i=1}^N Y_i(j)$. 
	
	In the spirit of classical optimal designs \citep{Atkinson2007}, we can define a design optimality criterion as a functional of the matrix  $\mathbf{V}_{\bm \tau}$ defined in (\ref{eq:cov_matrix}). For example, the D-optimality criterion, which aims to minimize the determinant of the covariance matrix, or the A-optimality criterion, which aims to minimize the trace of the covariance matrix, or the E-optimality criterion which aims to minimize the maximum eigenvalue of the covariance matrix, can be considered. 
	However, the second term $1/(N(N-1)) \sum_{i=1}^N  \left( \bm{\tau}_i  - \bm{\tau} \right) \left( \bm{\tau}_i  - \bm{\tau} \right)^{\T}$, which is a measure of heterogeneity of treatment effects, cannot be estimated from observed data, because none of the unit-level treatment effects $\bm{\tau}_i$ are estimable without additional assumptions due to the missing potential outcomes. Because $1/(N(N-1)) \sum_{i=1}^N  \left( \bm{\tau}_i  - \bm{\tau} \right) \left( \bm{\tau}_i  - \bm{\tau} \right)^{\T}$ is positive semi-definite, the first term of (\ref{eq:cov_matrix}) can be considered an upper bound of $\mathbf{V}_{\bm \tau}$, which is attained under specific restrictions on the potential outcomes (e.g., treatment effect homogeneity). Thus, we propose optimizing a functional of the first term of (\ref{eq:cov_matrix}), which in turn is equivalent to optimizing a functional of the positive definite matrix
	\begin{equation}
		\widetilde{\mathbf{V}} = \sum_{j=1}^{J} \frac{S^2_j}{N_j} \widetilde{\bm  \lambda}_j \widetilde{\bm \lambda_j}^{\T} =  \mathbf{L}^{\T} \mathbf{A} \mathbf{L}, \nonumber
	\end{equation}
	instead, where $\mathbf{A} = \text{diag}(S^2_1/N_1, \ldots, S^2_J/N_J)$.
	
	Another justification for using a functional of the matrix $\widetilde{\mathbf{V}}$ as an optimality criterion comes from the super-population perspective mentioned in Section \ref{ss:superpop}. \cite{DL2017SP} showed that the estimator $\widehat{\bm \tau}$ defined earlier is also an unbiased estimator of the super-population estimand $\bm{\tau}^{\text{SP}}$. Further, extending their argument for a single factor with two levels to the case of $2^K$ factorial designs, if $V_j^2 = \var[Y_i(j)]$, $j=1, \ldots, J$, then a variance decomposition yields the following sampling variance of $\widehat{\bm \tau}$  
	\begin{equation}
		\var^{\text{SP}} \left(\widehat{\bm \tau} \right) = \frac{1}{2^{2(K-1)}} \sum_{j=1}^{2^K} \frac{V_j^2}{N_j} \widetilde{\bm  \lambda}_j \widetilde{\bm \lambda_j} ^{\T} =  \frac{1}{2^{2(K-1)}} \sum_{j=1}^{2^K} \frac{E(S^2_j)}{N_j} \widetilde{\bm  \lambda}_j \widetilde{\bm \lambda_j} ^{\T}, \nonumber
	\end{equation}
	where $\var^{\text{SP}}$ denotes variance over the random sampling from the super population and random assignment, and $V_j^2 = E(S_j^2)$ represents expectation of $S^2_j$ with respect to the distribution of the potential outcomes in the super population.
	This connection provides further motivation for the form of our optimization, but we focus on the finite-population setting going forward.
	
	\subsection{Exact optimal designs} \label{ss:exact_CRD}
	
	The problem of finding an optimal design can be formulated as minimization of an appropriate functional $\psi(\widetilde{\mathbf{V}})$ subject to the constraint $\sum_{j=1}^J N_j = N$ or equivalently as $\sum_{j=1}^J p_j = 1$ in terms of the proportions of units $p_j = N_j/N$ to be assigned to treatment combination $j$. As discussed earlier, we consider three widely used functionals in optimal design literature: the D-optimality criterion where $\psi(\widetilde{\mathbf{V}}) = \left|\widetilde{\mathbf{V}} \right|$ and $\left|.\right|$ refers to the determinant, the A-optimality criterion where $\psi(\widetilde{\mathbf{V}}) = \textrm{tr} \left( \widetilde{\mathbf{V}} \right)$, and the E-optimality criterion where $\psi(\widetilde{\mathbf{V}}) = \max \{\nu_1, \ldots, \nu_J \}$ and $\nu_1, \ldots, \nu_J $ are the eigenvalues of $\widetilde{\mathbf{V}}$.  The following theorem, proved in Supplementary Material~\ref{App:ProofofResults}, summarizes these optimality results.
	
	\begin{theorem} \label{thm:optimal_allocation} Let $N$ units be allocated to $J$ treatment groups such that $p_j = N_j/N$ proportion of units receive treatment $j$. Then, the optimal allocation of $N$ units to $J$ treatment groups on the basis of covariance matrix  $\widetilde{\bm{V}}$ under  
		\begin{itemize}
			\item[(a)] A-optimality is  proportional to the finite-population standard deviations of potential outcomes in the corresponding treatment groups, i.e., $p_j = S_j/(\sum_j S_j)$.
			\item[(b)] D-optimality is balanced assignment to all $J$ treatment groups, i.e., $p_j = 1/J $.
			\item[(c)] E-optimality is proportional to the finite-population variances of potential outcomes in the corresponding treatment groups, i.e.,  $p_j = S_j^2/(\sum_j  S_j^2 )$.
		\end{itemize}
	\end{theorem}
	
	\medskip
	\begin{remark}
		The optimality results in Theorem~\ref{thm:optimal_allocation} are similar in spirit to the determination of optimal sample sizes in stratified survey sampling \citep[Ch. 5]{Cochran1977}. 
		The A-optimality result is the same as the Neyman-Allocation discussed in \cite{blackwell2022batch}, who motivate its use as reducing the identifiable portion of the finite-population variance of classical factorial estimators.
		Derivations of the A- and D- optimal designs are straightforward and use the Lagrangian multiplier-based optimization technique. The proof of the E-optimality result uses the idea of perturbing the eigenvalues of a scaled identity matrix to show that the E-optimal design is indeed characterized by equal eigenvalues of the matrix $\widetilde{\mathbf{V}}$.
	\end{remark}

	\begin{remark}
		In order to implement the results of Theorem~\ref{thm:optimal_allocation}, researchers need to ``guess'' the variances $S^2_j$, $j=1, \ldots, J$ and substitute them into the expressions for $p_j$. This is similar to the application of optimal designs in non-linear models, where optimal designs are actually ``locally optimal'' \citep{Chernoff1953}. It is often a common practice to conduct pilot studies to obtain some preliminary estimates of the $S^2_j$'s, as done in finite-population survey sampling.
	\end{remark}
	
	We now introduce two conditions associated with the matrix of potential outcomes under which Theorem~\ref{thm:optimal_allocation} can be further simplified.
	
	\begin{condition}[Homoscedasticity] \label{cond_homosc}
		We call an $N \times J$ matrix of potential outcomes \emph{homoscedastic} if each column has the same variance i.e., $S^2_j = S^2$ for $j=1, \ldots, J$. 
	\end{condition} 
	
	\begin{condition}[Strict additivity] \label{cond_str_add}
		Following \cite{DPR2015}, we call an $N \times J$ matrix of potential outcomes \emph{strictly additive} if $Y_i(j) - Y_i(\widetilde{j}) = \tau(j, \widetilde{j})$ for all $j \ne \widetilde{j} \in \{1, 2, \ldots, J\}$. Potential outcomes satisfying this condition also satisfy Condition~\ref{cond_homosc}.
	\end{condition} 
	
	The following corollary of Theorem~\ref{thm:optimal_allocation} is straightforward but useful:
	\begin{corollary} \label{cor:special1}
		If the matrix of potential outcomes satisfies Condition~\ref{cond_homosc}, then the A-, D- and E-optimal designs are all balanced designs with $p_j = 1/J$ for $j = 1, \ldots, J$.
		Further, under Condition~\ref{cond_str_add}, optimizations based on $\widetilde{\bm{V}}$ and $\mathbf{V}_{\bm \tau}$ are equivalent.
	\end{corollary}
	
	While Theorem~\ref{thm:optimal_allocation} provides results on exact optimal designs in terms of proportions $p_j$, experimenters need integer solutions in terms of $N_j$'s satisfying $\sum_j N_j = N$. For example, while the exact D-optimal design is balanced, $N$ is not necessarily a multiple of $2^K$, and the result does not provide a D-optimal allocation of, say, 69 units to the 8 treatment combinations in a $2^3$ factorial experiment. Thus we need approximate integer solutions to the optimization problem in which additional constraints on the sample sizes assigned to specific treatment groups can also be introduced. Next, we discuss an integer programming approach to obtain such solutions.
	
	\subsection{Computation of exact optimal designs using an integer programming approach} \label{ss:computation}
	
	Many sources in the integer programming literature address constrained optimization under integer space constraints \citep[e.g. see][]{nemhauser1988, schrijver1998, Khanthesis95, Sofi2020}. We adopt the methods proposed in \cite{FRIEDRICH20151} that are designed for settings very similar to the ones we consider. \cite{FRIEDRICH20151} consider the following integer programming problem:
	\begin{equation}
		\min_{N_1, \ldots, N_J} f(N_1, \ldots, N_J) \text{ s.t. } 
		\begin{cases}
			\sum_{j=1}^{J} N_j = N, \\ 
			l_j \le N_j \le u_j, \  \forall j=1,2, \ldots, J, \\
			N_j \in \mathbb{Z}_+^J, \  \forall j=1,2, \ldots, J.
		\end{cases} \label{eq:nlin_obj}
	\end{equation}
	where,
	$\mathbb{Z}_+$ is the set of positive integers, $(l_j,u_j)$ are the lower and upper bound constraints on $N_j$ and $f:\mathbb{R}_+^J \rightarrow \mathbb{R}$ is a convex function. If $f(N_1, \ldots, N_J)$ is separable (i.e., can be expressed as $\sum_{j=1}^J f_j(N_j)$) then the greedy algorithm in Figure~\ref{fig:greedyalg} finds the globally optimal integer solution of the minimization problem given in (\ref{eq:nlin_obj}) under some regularity conditions that we show in Supplementary Material \ref{conditionsGreedy}. 
	
	From the proof of Theorem~\ref{thm:optimal_allocation} in Supplementary Material~\ref{App:ProofofResults}, it follows that the A-optimality and D-optimality criteria can respectively be expressed as $\sum_{j=1}^J (S^2_j / N_j)$ and $\sum_{j=1}^J ( \log S^2_j / N_j)$. In Supplementary Material \ref{conditionsGreedy}, we show that the conditions for global convergence of the greedy algorithm are met  and thus, convergence to the true integer optimal solution in the cases of A- and D-optimality are guaranteed if this algorithm is used. Hence, substitution of $S^2_j/N_j$ and $\log S^2_j / N_j$ for $f_j(N_j)$ in the algorithm described in Figure~\ref{fig:greedyalg} leads to optimal integer solutions for the A-optimality criterion and the D-optimality criterion, respectively. In our implementation of this algorithm, we take $l_j = 2$ for $j=1, \ldots, J$ to guarantee at least two units are assigned to each treatment combination, allowing variance estimation within each treatment group.
	
	\begin{figure}[!htbp]
		\caption{Greedy algorithm for separable functions} \label{fig:greedyalg}
		\centering
		\fbox{
			\begin{minipage}{10 cm}
				\begin{enumerate}
					\item Set $I \leftarrow \{1,...,J\}$.
					\item Let $t=0$. Initialize $N_j^{(t)} \leftarrow l_j$ for $j=1, \ldots, J$. 
					\item While ($\sum_{j=1}^J N_j^{(t)} \ne N \ \& \ I \ne \phi$) do
					\begin{itemize}
						\item for all $(j \in I)$, $\delta_j \leftarrow f_j(N_j^{(t)}+1) - f_j(N_j^{(t)})$
						\item choose $\mathcal{J} \leftarrow \argmin_{j \in I} \delta_j$
						\item $j^* \leftarrow \min \mathcal{J}$
						\item if $N_{j^*}^{(t)} + 1 \le u_{j^*}$,  \mbox{then} $N_{j^*}^{(t+1)} \leftarrow N_{j^*}^{(t)} + 1$, $t \leftarrow t + 1$\\
						\mbox{else} \  $I \leftarrow I \setminus \{j^*\}$
					\end{itemize}
					\item Return optimal $(N_1, \ldots, N_J)$.
				\end{enumerate}
			\end{minipage}
		}
	\end{figure}
	
	We provide a modified greedy algorithm described in  Figure~\ref{fig:greedyalgE} for solving the E-optimality problem since the optimization problem cannot be written in a separable form as in A- or D-optimality criterion above. In this algorithm we take $f_j(N_j) = S_j^2/N_j$ and $l_j=2$. In Section \ref{ss:computation_RBD}, the ability of this greedy algorithm to find E-optimal solutions is demonstrated empirically for blocked designs, discussed next. 
	
	\begin{figure}[!htbp]
		\caption{Greedy algorithm for E-optimality} \label{fig:greedyalgE}
		\centering
		\fbox{
			\begin{minipage}{10 cm}
				\begin{enumerate}
					\item Set $I \leftarrow \{1,...,J\}$.
					\item Let $t=0$. Initialize $N_j^{(t)} \leftarrow l_j$ for $j=1, \ldots, J$. 
					\item While ($\sum_{j=1}^J N_j^{(t)} \ne N \ \& \ I \ne \phi$) do
					\begin{itemize}
						\item choose $\mathcal{J} \leftarrow \argmax_{j \in I} f_j(N_j^{(t)})$
						\item $j^* \leftarrow \min \mathcal{J}$			
						\item if $N_{j^*}^{(t)} + 1 \le u_{j^*}$,  \mbox{then} $N_{j^*}^{(t+1)} \leftarrow N_{j^*}^{(t)} + 1$, $t \leftarrow t + 1$ \\
						\mbox{else} \  $I \leftarrow I \setminus \{j^*\}$
					\end{itemize}
					\item Return optimal $(N_1, \ldots, N_J)$.
				\end{enumerate}
			\end{minipage}
		}
	\end{figure}

	\section{Optimal allocation for factorial experiments with blocks} \label{sec:RBDfactorial}
	Consider a block-randomized $2^K$ factorial design with $H$ blocks. 
	That is, units are pre-assigned membership to one of $h$ blocks based on some similarity metric (we do not consider how to form blocks here).
	Let $M_h$ denote the size of block $h$, $h = 1, \ldots, H$. Also, let $M_{h,j}$ be the number of units in block $h$ assigned to the treatment $j$ ($j = 1, \ldots, J$). Finally, let $b_i(h)$, $i=1, \ldots, N$ be an indicator variable taking value 1 if unit $i$ belongs to block $h$ and 0 otherwise.
	Treatment assignment under a factorial RBD is equivalent to performing an independent CRD, as described in Section~\ref{sec:CRDfactorial}, within each block.
	
	The population average treatment effect ${\bm \tau}$ can be expressed as $\sum_{h} M_h {\bm \tau}_h/ N$, where ${\bm \tau}_h$ is the block-level vector of factorial effects and its estimator $\widehat{\bm \tau}$ is a weighted average of $\widehat{\bm \tau}_h$, an unbiased estimator of ${\bm \tau}_h$ defined in the same way as in (\ref{eq:est_tau}) for block $h$. Extending (\ref{eq:cov_matrix}) by noting the independence of the assignment to treatment across blocks, the covariance matrix of $\widehat{\bm \tau}_h$ in a factorial RBD can thus be obtained as
	\begin{equation}
		\bm{V}_{{\bm \tau}_h} = \cov\left(\widehat{\bm \tau}_h \right) = \frac{1}{2^{2(K-1)}} \sum_{j=1}^{2^K} \frac{1}{M_{h,j}} \widetilde{\bm \lambda_j} \widetilde{\bm \lambda_j}^{\T} S^2_{h,j} - \frac{1}{M_h(M_h-1)} \sum_{i=1}^Nb_i(h) \left( \bm{\tau}_i  - \bm{\tau}_h \right) \left( \bm{\tau}_i  - \bm{\tau}_h \right)^{\T}, \label{eq:cov_matrix_blk}
	\end{equation}
	where $\widetilde{\bm \lambda}_j$ represents the transpose of row $j$ of the model matrix $\mathbf{L}$ defined in (\ref{eq:matrixL}) and $S^2_{h,j}$ denotes the variance of all $M_h$ potential outcomes for units in block $h$ under treatment $j$ with divisor $M_h-1$. The covariance matrix of $\widehat{\bm \tau}$ can then be expressed as $\cov \left( \sum_h M_h \widehat{\bm \tau}_h /N \right)$. Because the block-level treatment estimators $\widehat{\bm \tau}_h$ are independent across blocks, we have 
	\begin{align*}
		\cov\left( \widehat{\bm \tau} \right) &= \sum_{h=1}^H \frac{M_h^2}{N^2} \bm{V}_{{\bm \tau}_h}\\
		&= \frac{1}{2^{2(K-1)}} \sum_{j=1}^{2^K}  \widetilde{\bm \lambda_j} \widetilde{\bm \lambda_j}^{\T} \left[\sum_{h=1}^H \frac{M_h^2}{N^2}\frac{S^2_{h,j}}{M_{h,j}}\right] - \sum_{h=1}^H \frac{M_h^2}{N^2}\frac{1}{M_h(M_h-1)} \sum_{i=1}^Nb_i(h) \left( \bm{\tau}_i  - \bm{\tau}_h \right) \left( \bm{\tau}_i  - \bm{\tau}_h \right)^{\T}.
	\end{align*}
	
	Writing $S^2_{\text{blk}, j} = \sum_{h=1}^H (M_h^2/N^2)(S^2_{h,j}/M_{h,j})$ and proceeding along similar lines as in Section \ref{sec:CRDfactorial}, we formulate a surrogate optimization problem to only optimize the first term, since the second term in the equation above is not identifiable. Then, choosing $M_{h,j}$'s to optimize some functional of $\cov\left( \widehat{\bm \tau} \right)$ is equivalent to optimizing a functional of the matrix
	\begin{equation}
		\widetilde{\bm{V}}_{\text{blk}} = \sum_{j=1}^{2^K}  \widetilde{\bm \lambda_j} \widetilde{\bm \lambda_j}^{\T} S^2_{\text{blk}, j} =  \bm{L}^{\T} \bm{A}_{\text{blk}} \bm{L}, 
		\label{eq:blk_matrix_criterion}
	\end{equation}
	where $\bm{A}_{\text{blk}} = \text{diag}\left( S^2_{\text{blk}, 1}, \ldots, S^2_{\text{blk}, J} \right)$.
	
	\subsection{Exact optimal designs} \label{ss:exact_RBD}
	The problem of finding an optimal design for RBDs can be formulated as minimization of an appropriate functional $\psi(\widetilde{\bm{V}}_{\text{blk}} )$ subject to the constraint $\sum_{j=1}^J M_{h,j} = M_h$ for each $h$ or equivalently as $\sum_{j=1}^J p_{h,j} = 1$ in terms of the proportions of units to be assigned to treatment combination $j$ in block $h$, $p_{h,j} = M_{h,j}/M_h$. While the A-optimality result is straightforward, finding exact D-optimal and E-optimal solutions in the setting with blocks is difficult without imposing restrictions on the potential outcomes. Before stating the optimality results, we first introduce two such restrictions that generalize Condition~\ref{cond_homosc} to a block setting.
	
	\begin{condition}[Within-block homoscedasticity, WBH] \label{cond_WBH}
		We call an $N \times J$ matrix of potential outcomes in $H$ blocks to be \emph{within-block homoscedastic (WBH)} if within each block, all treatment columns have the same variance, i.e., within block $h$, $S^2_{h,j} = S^2_{h \cdot}$ for $j=1, \ldots, J$.
	\end{condition} 
	
	\begin{condition}[Between-block homoscedasticity, BBH] \label{cond_BBH}
		We call an $N \times J$ matrix of potential outcomes in $H$ blocks to be \emph{between-block homoscedastic (BBH)} if for each treatment column $j$, the variance of potential outcomes in each block is the same, i.e., $S^2_{h,j} = S^2_{\cdot, j}$ for $h=1, \ldots, H$ and $j=1,\ldots,J$. 
	\end{condition} 
	
	The following theorem now summarizes the optimality results for blocked designs.
	
	\begin{theorem} \label{thm:optimal_allocation_blk} Let $N$ units be distributed across $H$ blocks such that there are $M_h$ units in block $h$ and $N = \sum_{h=1}^H M_h$. Let each set of $M_h$ units be allocated to $J$ treatment groups such that $p_{h,j}$ proportion of the units are allocated to treatment $j$ in block $h$, and let $S^2_{h,j}$ as defined in (\ref{eq:cov_matrix_blk}). Then, optimal allocation of $M_h$ units to the $J$ treatment groups on the basis of covariance matrix $\widetilde{\bm{V}}_{\text{blk}}$ under different optimality criteria can be summarized as follows.
		\begin{itemize}\setlength{\itemindent}{-0.7em}
			\item[(a)] The A-optimal allocation is the same as the A-optimal CRD allocation within each block, i.e., $p_{h,j} = S_{h,j}/(\sum_{j=1}^{J} S_{h,j})$ for each $h$.
			\item[(b)] If either (or both) of Conditions \ref{cond_WBH} (WBH) and \ref{cond_BBH} (BBH) hold, the D-optimal allocation is the balanced assignment within each block, i.e., $p_{h,j} = 1/J$ for each $h$.
			\item[(c)] If Condition~\ref{cond_WBH} (WBH) holds, the E-optimal allocation is the balanced design within each block, i.e., $p_{h,j} = 1/J$ for each $h$.
		\end{itemize}
	\end{theorem}
	
	\medskip
	
	\begin{remark} \label{remark:str_add_blk1}
		Condition~\ref{cond_str_add} for all $N$ units implies both WBH and BBH. Consequently, by Theorem~\ref{thm:optimal_allocation_blk}, the D- and E-optimal allocation for strictly additive potential outcomes in randomized block designs is a balanced assignment within each block.
	\end{remark}
	
	\begin{remark} \label{remark:str_add_blk2}
		The A-optimal allocation under WBH is balanced allocation within each block (same as D- and E-optimal allocations). However, under BBH, the A-optimal allocation is different from the D- and E-optimal allocations, and is proportional to the standard deviations of the treatments $S_{\cdot, j}$ that is constant across blocks.
	\end{remark}
	
	\subsection{Computation of exact optimal designs for factorial RBDS using an integer programming approach} \label{ss:computation_RBD}
	
	As in the case of completely randomized factorial designs, whereas Theorem~\ref{thm:optimal_allocation_blk} provides results on exact optimal designs in terms of proportions $p_{h,j}$, experimenters  need integer solutions in terms of $M_{h,j}$'s in which additional constraints on the sample sizes assigned to specific treatment groups can also be introduced. Further, Theorem~\ref{thm:optimal_allocation_blk} provides D- and E-optimal solutions only under specific conditions like WBH and BBH. Thus, we discuss an integer programming approach to obtain integer solutions for settings covered and not covered by Theorem~\ref{thm:optimal_allocation_blk}.
	
	We can use the same algorithm in Figure~\ref{fig:greedyalg} within each block to obtain the optimal integer solutions for  A-optimality under the RBD by replacing the function $f_j(.)$ by $f_{h,j}(M_{h,j})= (M_h^2/N^2)(S_{h,j}^2/M_{h,j})$. For D- and E-optimality, we extend the greedy idea from  the algorithm in Figure~\ref{fig:greedyalgE} with minor changes to the function $f_j(.)$. We take $f_{h,j}(M_{h,j})= \log(\sum_{h=1}^H(M_h^2/N^2)(S_{h,j}^2/M_{h,j}))$ for D-optimality and $f_{h,j}(M_{h,j})= (M_h^2/N^2)(S_{h,j}^2/M_{h,j})$ for E-optimality.
	The exact algorithms taking the structure of the blocks into account are given in Supplementary Material~\ref{append:greed_alg}.
	The main difference between the algorithm used in Section \ref{sec:CRDfactorial} and the one proposed here lies in the fact that now we have to allocate the next best unit at iteration $t$ over an $H \times J$ matrix $((M_{h,j}^{(t)}))$ with upper bounds on row sums, rather than a vector $(N_1^{(t)},\dots,N_J^{(t)}$) with an upper bound on the sum of the elements, as in the case of the CRD. Note that, by \cite{FRIEDRICH20151}, the greedy algorithm finds the correct solution in the case of A-optimality for the block design. It, however, does not extend to the D- and E-optimality in the block case, due to nature of the objective functions.
	
	We now conduct an empirical exploration of the performances of the greedy algorithm in terms of its ability to find E-optimal solutions. Five different settings of $2^2$ factorial designs in two blocks, each corresponding to a specific type of potential outcome matrix, are considered. Each setting is defined by the block sizes $M_1$ and $M_2$, and the $4 \times 2$ matrix of variances $S^2_{h,j}$ as shown in columns 3 and 4 of Table~\ref{tbl:E-opt_GreedySim}, respectively.
	
	The first setting considers blocks of equal sizes with potential outcomes satisfying Condition~\ref{cond_str_add} (strict additivity), leading to an E-optimal design that is balanced within each block as per Remark~\ref{remark:str_add_blk1}. The second setting considers equal block sizes with potential outcomes satisfying Condition~\ref{cond_WBH} (WBH), leading to a balanced design by Theorem~\ref{thm:optimal_allocation_blk}.  In this setting and the previous one, the exact optimal designs provide optimal integer solutions. This is not the case in the third setting, which considers unequal block sizes with potential outcomes satisfying Condition~\ref{cond_BBH} (BBH) but not Condition~\ref{cond_WBH} (WBH). Theorem~\ref{thm:optimal_allocation_blk} does not apply directly for E-optimality, but the greedy algorithm identifies the unique true E-optimal allocation determined by the exhaustive search. The fourth setting is similar the third case above, but is one where the exhaustive search provides multiple solutions, identifying four different allocations, each of which is optimal. In this case, Theorem~\ref{thm:optimal_allocation_blk} does not apply directly and the greedy algorithm identifies one of these solutions. The fifth setting neither satisfies Condition~\ref{cond_WBH} (WBH) nor Condition~\ref{cond_BBH} (BBH), and consequently Theorem~\ref{thm:optimal_allocation_blk} cannot provide an exact E-optimal solution. However, the greedy algorithm identifies one of the two (identified through exhaustive search) true optimal integer allocations. A quick note on our greedy algorithm is that, due to the nature of the algorithm in Figures~\ref{fig:greedyalgBlk} and \ref{fig:greedyalgEblk} of Supplementary Material~\ref{append:greed_alg}, ties are broken deterministically using minimum index when the greedy step returns more than one solution. Thus, our greedy solutions will always achieve the same solution for a given set of inputs without regard for the plurality of solutions (such as the ones in the fourth and fifth setting above).
	
	A similar exploration performed for the D-optimal allocation (shown in Supplementary Material~\ref{App: EmpEvidenceGreedyBlk}) provides evidence that the greedy algorithm can identify the true optimal integer solution when it is unique, and \emph{one} of the true optimal solutions when multiple optimal integer solutions exist.
	
	\begin{table}[ht]
		\centering \footnotesize
		\begin{tabular}{c|m{4cm}|m{1cm}|m{2.5cm}|m{3.5cm}|m{3.5cm}}
			S.No. & Case & Block Size ($M_h$) & Variances ($S_{h,j}^2$)  & Exhaustive search \newline E-optimal solution & Greedy solution \\[3ex]
			\hline & & & & & \\
			1. & Equal blocks with equal variances & 
			$\begin{bmatrix} 40 \\40 \end{bmatrix}$ & 
			$\begin{bmatrix} 1 & 1 & 1 & 1 \\ 1 & 1 & 1 & 1 \end{bmatrix}$ &  $\begin{bmatrix} 10 & 10 & 10 & 10 \\ 10& 10 & 10 & 10 \end{bmatrix} $ & $\begin{bmatrix} 10 & 10 & 10 & 10 \\ 10& 10 & 10 & 10 \end{bmatrix}$ \\ 
			& & & & & \\
			2. & Equal blocks with equal variances for all treatments within block & 
			$\begin{bmatrix} 40 \\40 \end{bmatrix}$ & 
			$\begin{bmatrix} 4 & 4 & 4 & 4 \\ 1 & 1 & 1 & 1 \end{bmatrix}$ &  $\begin{bmatrix} 10 & 10 & 10 & 10 \\ 10& 10 & 10 & 10 \end{bmatrix} $ & $\begin{bmatrix} 10 & 10 & 10 & 10 \\ 10& 10 & 10 & 10 \end{bmatrix}$ \\
			& & & & & \\
			3. & Unequal blocks with equal variances across blocks for each treatment & 
			$\begin{bmatrix} 40 \\20 \end{bmatrix}$ & 
			$\begin{bmatrix} 1 & 2 & 3 & 4 \\ 1 & 2 & 3 & 4 \end{bmatrix}$ &  $\begin{bmatrix} 4 & 8 & 12 & 16 \\ 2 & 4 & 6 & 8 \end{bmatrix} $ & $\begin{bmatrix} 4 & 8 & 12 & 16 \\ 2 & 4 & 6 & 8 \end{bmatrix}$ \\
			& & & & & \\
			4. & Unequal blocks with equal variances but exact solution is non-integer & 
			$\begin{bmatrix} 40 \\20 \end{bmatrix}$ & 
			$\begin{bmatrix} 1 & 2 & 3 & 5 \\ 1 & 2 & 3 & 5 \end{bmatrix}$ &  $\begin{bmatrix} 
				\begin{bmatrix} 4 & 8 & 11 & 17\\2& 3 & 5 & 10 \end{bmatrix} \\[1em]
				\begin{bmatrix} 4 & 7 & 11 & 18\\2& 4 & 5 & 9 \end{bmatrix}\\[1em]
				\begin{bmatrix} 3 & 8 & 11 & 18\\3& 3 & 5 & 9 \end{bmatrix}\\[1em]
				\begin{bmatrix} 3 & 7 & 11 & 19\\3& 4 & 5 & 8\end{bmatrix}
			\end{bmatrix}$ & 
			$\begin{bmatrix} 4 & 7 & 11 & 18\\2& 4 & 5 & 9 \end{bmatrix}$ \\
			& & & & & \\
			5. & Equal blocks with unequal variances & 
			$\begin{bmatrix} 40 \\40 \end{bmatrix}$ & 
			$\begin{bmatrix} 1 & 2 & 3 & 4 \\ 4& 3 & 2& 1 \end{bmatrix}$ &  $\begin{bmatrix} 
				\begin{bmatrix} 6 & 10 & 11 & 13 \\ 13 & 11 & 10 & 6 \end{bmatrix} \\[1em]
				\begin{bmatrix} 6 & 9 & 12 & 13 \\ 13 & 12 & 9 & 6 \end{bmatrix}
			\end{bmatrix} $ & $\begin{bmatrix} 6 & 9 & 12 & 13 \\ 13 & 12 & 9 & 6 \end{bmatrix}$ 
		\end{tabular}
		\caption{Summary of Greedy algorithm solutions for E-optimality for $H=2, K=2$}
		\label{tbl:E-opt_GreedySim}
	\end{table}
	
	\section{Optimal allocation driven by cost constraints} 
	\label{sec:costbased_optimal_alloc}
	
	So far, we have considered optimality criteria that are based on the covariance matrix of the estimated factorial effects, implicitly assuming that  all treatment combinations are equally expensive (with respect to cost and/or time). However, such assumptions may not be true in many practical situations and cost constraints can play an important role in determining optimal allocation. Thus it is worthwhile to explore solutions to optimal allocation under cost constraints.
	
	We consider the optimal allocation for $2^K$ factorial CRDs. Let the cost of assigning treatment combination $j$ to one unit be $C_j>0$, and the total available budget be $C$. In the new optimization problem, we replace the constraint $\sum_j N_j = N$ in the original problem described in Section \ref{ss:exact_CRD} by the cost constraint $\sum_j C_j N_j \le C$. The new optimization problem is therefore: 
	\begin{eqnarray}\label{eq:cost_problem}
		\min_{N_j}  \psi\left( \widetilde{\mathbf{V}} \right) \text{subject to } \sum_j C_jN_j \le C,
	\end{eqnarray}
	where $\widetilde{\mathbf{V}} =\sum_{j=1}^{J} \frac{S^2_j}{N_j} \widetilde{\bm  \lambda}_j \widetilde{\bm \lambda_j}^{\T}$ and $\psi(\widetilde{\mathbf{V}})$ is a functional of $\widetilde{\mathbf{V}}$. A straightforward approach to incorporate this new constraint into our previous setting is to re-write the constraint as $\sum_j \widetilde N_j \le C$ where $\widetilde N_j = N_jC_j$ is the total cost for the suggested allocation to treatment arm $j$. Under this one-to-one transformation $\tilde N_j = C_jN_j$, the  optimization problem in (\ref{eq:cost_problem}) is equivalent to minimizing the objective function over $\widetilde N_j$ \citep{boyd2004convex}, and can be written as:
	\begin{eqnarray*}
		\min_{\tilde N_j}  \psi\left( \widetilde{\mathbf{V}} \right) &=& \min_{\tilde N_j}  \psi \left( \sum_{j=1}^{J} \frac{S^2_j}{N_j} \widetilde{\bm  \lambda}_j \widetilde{\bm \lambda_j}^{\T} \right) = \min_{\tilde N_j} \psi \left(   \sum_{j=1}^{J} \frac{S^2_j}{\tilde N_j/C_j} \widetilde{\bm  \lambda}_j \widetilde{\bm \lambda_j}^{\T} \right), \label{eq:cost_opt1} \\
		&& \mbox{subject to} \ \sum_j \tilde{N}_j \le C. \nonumber
	\end{eqnarray*}
	
	Because the optimal solution of the above optimization problem is attained at $\sum_j \widetilde{N}_j = C$, the inequality constraint can be replaced by the equality constraint.   Then, proceeding along the lines of Theorem~\ref{thm:optimal_allocation}, one can obtain the cost for the optimal allocation to treatment arm $j$ as $\widetilde{N}_j \propto S_j \sqrt(C_j)$, $\tilde N_j = C/J$ and $\widetilde{N}_j \propto S_j^2 C_j$ as the A-, D- and E- optimal solutions. These results are formalized in terms of the optimal proportion of the budget allocated to each treatment arm, which can be used to determine the number of units to assign to each treatment arm, in the theorem below.
	
	\begin{theorem} \label{thm:cost_opt_allocation}
		Let $C$ be total budget for the whole experiment and the cost of allocating one experimental unit to treatment $j$ be  $C_j>0$. Let $\pi_j = C_jN_j/C$ denote the proportion of the total budget assigned to treatment $j$ with $\sum_j \pi_j \le 1$.
		Then, the
		\begin{itemize}
			\item[(a)] A-optimal cost-based allocation to the $J$ treatment groups on the basis of covariance matrix $\widetilde{\bm{V}}$ is $\pi_j =(S_j \sqrt C_j)/(\sum_j S_j \sqrt C_j)$.
			\item[(b)] D-optimal cost-based allocation to the $J$ treatment groups on the basis of covariance matrix $\widetilde{\bm{V}}$ is $\pi_j = 1/J$.
			\item[(c)] E-optimal cost-based allocation to the $J$ treatment groups on the basis of covariance matrix $\widetilde{\bm{V}}$ is $\pi_j = (S_j^2C_j)/(\sum_j S_j^2 C_j)$. 
		\end{itemize}
	\end{theorem}
	
	\begin{remark} \label{rem:cost_block}
		Theorem~\ref{thm:cost_opt_allocation} can be extended to the case of block designs along the lines of Theorem~\ref{thm:optimal_allocation_blk}.
	\end{remark}
	
	\begin{remark} \label{rem:thms1_3}
		If for $j=1, \ldots, J$, the costs $C_j$ in Theorem~\ref{thm:cost_opt_allocation} are the same and equal to $C_0$, then the constraint $\sum_j C_j N_j \le C$ reduces to $\sum_j N_j \le N$ where $N = C/C_0$. Also, $\pi_j = C_0 N_j/C \equiv p_j$. Thus the optimization problem becomes the same as the one in Theorem \ref{thm:optimal_allocation}, making it a special case of Theorem \ref{thm:cost_opt_allocation}.
	\end{remark}
	
	We use an example to demonstrate the applicability of Theorem \ref{thm:cost_opt_allocation}. Consider a $2^2$ factorial design, with $C=100$, per unit cost vector $(C_1, \ldots, C_4) = (0.1,4,4,9)$. This set up represents many common scenarios where treatment arm 1 represents the control group 00 and involves a per-unit cost that is negligible compared to the ones with at least one active treatment. On the other hand, treatment arm 4 has both treatments at active level and involves the highest cost. Table \ref{tab:thm3demo} shows the A-, D- and E- optimal proportions of total cost $\pi_j$'s for two different vectors of variances $(S_1^2, \ldots, S_4^2)$. In one setting, we take the vector as $(1,1,1,1)$ and in another, set it to $(1,2,3,4)$. For the sake of completeness, we also add a column of equal cost $(1,1,1,1)$, under which the $p_j$'s of Theorem \ref{thm:optimal_allocation} and $\pi_j$'s of Theorem \ref{thm:cost_opt_allocation} become identical, as explained in Remark \ref{rem:thms1_3}. Thus, the optimal allocations in the first column of Table \ref{tab:thm3demo} can also be derived from Theorem \ref{thm:optimal_allocation} with $N = C = 100$.
	
	\begin{table}[!htbp]
		\centering \small
		\caption{Optimal $\pi_j$'s under cost constraints obtained from Theorem \ref{thm:cost_opt_allocation} } \label{tab:thm3demo}
		\begin{tabular}{|c|c|c|c|} \hline
			Variance & Type of &  \multicolumn{2}{c|}{Cost vector $(C_1, \ldots, C_4)$} \\ 
			Vector & Optimality &   $(1,1,1,1)$ & $(0.1,4,4,9)$  \\ \hline \hline
			& A &  (0.250,0.250,0.250,0.250) & (0.043,0.273,0.273,0.410) \\ \cline{2-4}
			(1,1,1,1) & D  & (0.250,0.250,0.250,0.250) & (0.250,0.250,0.250,0.250) \\ \cline{2-4}
			& E &  (0.250,0.250,0.250,0.250) & (0.006,0.234,0.234,0.526)  \\ \hline \hline
			& A &   (0.163,0.230,0.282,0.325) & (0.025,0.224.0.275,0.476) \\ \cline{2-4}
			(1,2,3,4) & D  & (0.250,0.250,0.250,0.250) & (0.250,0.250,0.250,0.250) \\ \cline{2-4}
			& E  & (0.100,0.200,0.300,0.400) & (0.002,0.143,0.214,0.642)  \\ \hline
		\end{tabular}
	\end{table}
	
	One can obtain the number of units $N_j$'s the A-, D- and E-optimal allocations of $N_j$'s by substituting the optimal $\pi_j$s from Theorem \ref{thm:cost_opt_allocation} into $N_j = (C \pi_j)/C_j$. However, rounding these optimal $N_j$'s into nearest integers may lead to violation of the constraint $\sum_j C_j N_j \le C$. To avoid such possibilities, one can consider the optimal values of $\lfloor{(C \pi_j)/C_j} \rfloor$ as approximate integer solutions, where $\lfloor{x}\rfloor$ denotes the largest integer contained in $x$.
	
	Researchers may decide to impose an additional constraint on the optimization problem (\ref{eq:cost_opt1}) that forces the sum of $N_j$'s to be exactly equal to a predetermined $N$. Such a problem would give optimal allocation under fixed $N$, unlike Theorem~\ref{thm:cost_opt_allocation}. However, imposing this additional constraint may force the set of feasible solutions to the optimization problem to be empty. For example, suppose for all $j$, $C_j > C/N$. Then, $\sum_j C_jN_j>\sum_j (C/N)N_j$, which exceeds the allowable cost $C$ if the restriction $\sum_j N_j = N$ is imposed. Thus, additional conditions are necessary to guarantee that the feasible set is non-empty. Obtaining closed-form solutions under such conditions may not be straightforward and one may need to rely on numerical methods to obtain such solutions.
	
	\section{Applications in real experiments} \label{sec:examples}
	
	In this section, we demonstrate applications of the results and algorithms developed to two real-life experiments. First, we re-visit the education example from \cite{Angrist2009} described in Section \ref{sec:intro}. Second, we discuss a pilot audit experiment reported in \cite{Libgober2020} conducted to identify how perceptions of race, gender and affluence affect access to lawyers, and demonstrate how the proposed methodology can be used to design follow-up experiments in similar populations.
	
	\subsection{Education experiment} \label{sec:example_ed}
	
	In the experiment described in \cite{Angrist2009}, the authors use a CRD to allocate the $N=1656$ units to the $2^2$ treatments. Theorem \ref{thm:optimal_allocation} can directly inform us of the optimal allocation without costs, but there is more structure that we can exploit. For instance, there are potentially two blocks of experimental units or subjects representing female (block 1) and male (block 2) students, with block sizes $M_1 = 948$ and $M_2 = 708$, which can be used to improve their design. Theorem~\ref{thm:optimal_allocation_blk} will give us the optimal designs in this case. 
	
	Assuming that there is no prior information about the variances of potential outcomes (GPAs after year 1), we assume that the variances are equal within and across blocks (Conditions~\ref{cond_WBH} and~\ref{cond_BBH}). Then, optimal allocations under both CRD and RBD, from Theorem~\ref{thm:optimal_allocation} and Theorem~\ref{thm:optimal_allocation_blk} respectively, are shown in Tables~\ref{tab:Angrist_Results1a} and~\ref{tab:Angrist_Results1b}.

	\begin{table}[!htbp]
		\centering \small
		\caption{A-, D- \& E-optimal allocations under CRD assuming Condition~\ref{cond_homosc}} \label{tab:Angrist_Results1a}
		\begin{tabular}{|c|c|c|c|c|} \hline
			& \multicolumn{4}{c|}{Treatment combination} \\ 
			& $00$ & $01$ & $10$ & $11$ \\ \hline 
			$N$= 1656 & 414 & 414 & 414 & 414 \\ \hline
		\end{tabular}
	\end{table}
	
	\begin{table}[!htbp]
		\centering \small
		\caption{A-, D- \& E-optimal allocations under RBD assuming Conditions~\ref{cond_WBH} \&~\ref{cond_BBH}} \label{tab:Angrist_Results1b}
		\begin{tabular}{|c|c|c|c|c|c|} \hline
			& Block& \multicolumn{4}{c|}{Treatment combination} \\ 
			Block&size  & $00$ & $01$ & $10$ & $11$ \\ \hline 
			1 &$M_1=$ 948& 237 & 237 & 237 & 237 \\
			2 &$M_2=$ 708&  177 & 177 & 177 & 177 \\ \hline
		\end{tabular}
	\end{table}

	Now let us consider a hypothetical situation where the number of units $N$ is not prespecified, but there is a budget constraint that depends on the costs associated with the four treatment combinations in this experiment. The treatment combinations $01$ (SFP but not SSP) and $10$ (SSP but not SFP), each involve cost associated with one of two programs. \cite{Angrist2009} report about \$5,000 for individual students that were allocated to treatment $10$ (SSP). Per unit cost for treatment combination $01$(SFP) is not mentioned but if we assume a similar cost as with SSP, then, we can infer that the cost to allocate a student to the treatment combination $11$ (SFSP) would be the sum total of the individual costs (\$10,000). The control, representing the treatment combination $00$, is possibly the cheapest to allocate units to, because it would involve only administrative cost, which we assume to be \$500. Then under the original allocation (1106, 250, 250, 150) in the actual experiment as shown in Table \ref{tab:Angrist}, the cost of the experiment would be approximately \$4.5 million. Assuming this amount to be our budget constraint $C$, the A-, D- and E-optimal allocations for two different variance vectors obtained from Theorem \ref{thm:cost_opt_allocation} are shown in Table \ref{tab:Angrist_Results2}. The first row shows the allocation of the proportions $\pi_j$'s of the total budget to the four treatment arms, and the second shows the corresponding approximate integer solution for $N_j$ as $\lfloor{C \pi_j/C_j}\rfloor$.
	
	\begin{table}[!htbp]
		\centering \small
		\caption{Optimal allocations ($\pi_j$ and $N_j = \lfloor{C \pi_j/C_j}\rfloor$) with cost vector (500,5000,5000,10000) and total budget of $C=\$4.5$ million} \label{tab:Angrist_Results2}
		\begin{tabular}{|c|c|c|c|c|c|c|} \hline
			Variance vector & Type of & & \multicolumn{4}{c|}{Treatment combination} \\ 
			$(S_1^2, \ldots, S_4^2)$ & Optimality & &  $00$ & $01$ & $10$ & $11$ \\ \hline 
			&A & $\pi_j$& 0.085 & 0.268 & 0.268 & 0.379 \\ 
			\cline{3-7}
			&&$N_j$  & 762 & 241 & 241 & 170 \\
			\cline{2-7}
			(1,1,1,1) & D &$\pi_j$  & 0.25 & 0.25 & 0.25 & 0.25 \\ 
			\cline{3-7}
			&&$N_j$  & 2250 & 225 & 225 & 112 \\
			\cline{2-7}
			& E & $\pi_j$& 0.024 & 0.244 & 0.244 &0.488
			\\
			\cline{3-7}
			&&$N_j$  & 219 & 219 & 219 & 219 \\ \hline
			&A & $\pi_j$&  0.062 & 0.275 & 0.275 & 0.389 \\ 
			\cline{3-7} 
			&&$N_j$  & 553 & 247 & 247 & 174 \\
			\cline{2-7}
			(1,2,2,2) & D &$\pi_j$ & 0.25 & 0.25 & 0.25 & 0.25 \\ 
			\cline{3-7}
			&&$N_j$  & 2250 & 225 & 225 & 112 \\ \cline{2-7}
			& E &$\pi_j$ & 0.012 & 0.245 & 0.245 & 0.494 \\
			\cline{3-7}
			&&$N_j$  & 111 & 222 & 222 & 222 \\ \hline
		\end{tabular}
	\end{table}
	
	\subsection{Audit experiment}
	
	\cite{Libgober2020} reported an audit study in which the experimental units were 96 lawyers randomly selected from lawyers in California with a certification in criminal law.  Each lawyer in the experiment received an email about a routine `driving under influence' (DUI) case (a very common criminal matter). The email template suggested that the person sending the email was (i) either white or black (with a racially distinctive name being used to influence perceived race), (ii) either female or male (again cued via the email sender's name), and (iii) either relatively affluent or relatively lower-income description of client's earnings. Thus, this experiment had a $2^3$ factorial structure. The response was recorded as a binary outcome taking value 1 if there was a response to the email and 0 otherwise. The experiment was replicated with 96 additional lawyers after a certain period of time. The estimated variances $s^2_j$ for $j=1, \ldots, 8$ treatment groups for the individual replicates and their pooled values are shown in Table~\ref{tab:Libgober_varest}.
	
	\begin{table}[htbp]
		\centering \small
		\caption{Estimated variances for different treatment groups} \label{tab:Libgober_varest}
		\begin{tabular}{|c|c|c|c|c|c|c|c|c|} \hline
			Experiment & $000$ & $001$ & $010$ & $011$ & $100$ & $101$ & $110$ & $111$ \\ \hline
			Replicate I & 0.15 & 0.15 & 0.15 & 0.20 & 0.27 & 0.15 & 0.27 & 0.27 \\ \hline 
			Replicate II & 0.27 & 0.24 & 0.20 & 0.20 & 0.20 & 0.27 & 0.27 & 0.15 \\ \hline
			Pooled & 0.21 & 0.20 & 0.18 & 0.20 & 0.23 & 0.21 & 0.27 & 0.21 \\ \hline
		\end{tabular}
	\end{table}
	
	If another completely randomized experiment is planned with lawyers selected from a similar pool with a sample size of 192, then based on the pooled estimated variances shown in Table~\ref{tab:Libgober_varest}, we can apply Theorem~\ref{thm:optimal_allocation} to obtain the optimal designs given in Table~\ref{tab:Libgober_CRD}.
	
	\begin{table}[htbp]
		\centering \small
		\caption{Optimal allocations for future CRD} \label{tab:Libgober_CRD}
		\begin{tabular}{|c|c|c|c|c|c|c|c|c|} \hline
			Optimality & $000$ & $001$ & $010$ & $011$ & $100$ & $101$ & $110$ & $111$ \\ \hline
			A & 24 & 23 & 22 & 23 & 25 & 24 & 27 & 24 \\ \hline 
			D & 24 & 24 & 24 & 24 & 24 & 24 & 24 & 24 \\ \hline
			E & 24 & 22 & 20 & 22 & 26 & 24 & 30 & 24\\ \hline
		\end{tabular}
	\end{table}
	
	Now assume for illustration that (contrary to fact) instead of two replicates, the original experiment was conducted in two blocks, each block representing one type of lawyer (e.g., criminal and divorce), and suppose we want to obtain optimal allocations within each block for a future experiment. Further suppose that the variance estimates in row $j$ of Table~\ref{tab:Libgober_varest} represent estimates of the variances $S^2_{h,j}$ in block $j=1,2$ and block 2 respectively. Then, we can directly use part (a) of Theorem~\ref{thm:optimal_allocation_blk} to derive the A-optimal design. However, neither WBH or BBH appear to hold, and we cannot apply parts (b) and (c) of Theorem~\ref{thm:optimal_allocation_blk}. Conveniently, we can obtain the D- and E-optimal designs using the greedy search algorithm proposed in Section \ref{ss:computation_RBD}.

	\begin{table}[htbp]
		\centering \small
		\caption{Optimal allocations for future RBD} \label{tab:Libgober_RBD}
		\begin{tabular}{|c|c|c|c|c|c|c|c|c|c|} \hline
			Optimality & Block &  $000$ & $001$ & $010$ & $011$ & $100$ & $101$ & $110$ & $111$ \\ \hline
			A & I & 11 & 11 & 10 & 12 & 14 & 10 & 14 & 14  \\
			& II & 13 & 13 & 12 & 11 & 11 & 13 & 13 & 10 \\ \hline
			D &  I & 11 & 11 & 12 & 13 & 13 & 10 & 12 & 14 \\ 
			& II & 13 & 13 & 13 & 12 & 11 & 13 & 11 & 10   \\\hline
			E &  I & 10 & 10 & 10 & 12 & 15 & 10 & 16 & 13   \\ 
			& II & 13 & 12 & 10 & 11 & 12 & 13 & 15 & 10  \\ \hline
		\end{tabular}
	\end{table}

	\section{Discussion}  \label{sec:discussion}
	
	In this paper, we consider optimal allocations of a finite population of experimental units to different treatment combinations of a $2^K$ factorial experiment under the potential outcomes model. Rather than invoking the standard assumption in the mainstream optimal design literature that outcome data comes from a known family of distributions, our work revolves around randomization-based causal inference for the finite-population setting. We find that for $2^K$ factorial designs with a completely randomized treatment assignment mechanism, D-optimal solutions are always balanced designs, while A- and E-optimal solutions are proportional to finite-population standard deviations and finite-population variances of the treatment groups, respectively. For blocked designs, our solution does not admit a closed form for D- or E-optimality without imposing specific restrictions on the potential outcomes, but the A-optimal allocation is equivalent to finding the A-optimal solution within each block. Convenient integer-constrained programming solutions using a greedy optimization approach to find integer optimal allocation solutions for both complete and block randomization are proposed. Optimal allocations are also derived under cost constraints.
	
	While there is a large literature on model-based optimal designs, to the best of our knowledge, such designs have had very limited development for randomization-based inference for finite populations. The ideas explored and results developed in this paper exploit the connection between finite-population sampling and experimental design. This recondite connection has recently been emphasized, explored, and utilized in various contexts by several researchers in causal inference, as discussed in \cite{MDR2018}. This article attempts to further strengthen the bridge between finite-population survey sampling and experimental design by utilizing ideas from proportional and optimal allocation for stratified sampling in the context of optimal designs. While the optimal solutions are derived for a finite-population setting, they are readily applicable to a super-population setting without making any assumptions about the probability distribution of the outcome variable.  
	
	A question that practitioners may ask is, which optimal design should be chosen for a given experiment? The answer would depend on the research goal of the experimenter. As our results have shown, strong assumptions like strict additivity lead to equivalence of A-, D- and E- optimal designs. However, under treatment effect heterogeneity, different criteria will lead to different allocations. Both A- and D-optimality criteria are associated with quality of estimated causal effects - whereas A-optimality minimizes the average variance of estimators, the D-optimality criterion minimizes the volume of the confidence ellipsoid around the parameters. Some researchers \citep[e.g.,][]{JonesMoyerGoos2020} have argued that in model-based settings, A-optimal designs exhibit better performance than D-optimal designs when the objective is screening of active effects from inactive ones. On the other hand, when the goal is to draw the most precise inference on the vector of estimated causal effects, D-optimal design may be a better choice. The goal of the E-optimal design is to minimize the maximum variance of all possible normalized linear combinations of estimated treatment effects. Thus the E-optimal design is useful when a large number of linear combinations of factorial effects are of interest. The E-optimal allocation, being a minimax strategy, is likely to provide a more conservative solution to the inference problem, but as shown by some researchers \citep[e.g.,][]{WKW1994} in other contexts, the E-optimal solution may be less sensitive to incorrect prior information or assumptions about potential outcomes in comparison to A- and D-optimal designs. However, more investigation is required along these lines in the randomization-based setting.
	
	The work presented in this paper can be extended in several directions. One limitation of the proposed approach lies is the fact that the correlation among the potential outcomes under different treatment combinations is unidentifiable from the data, forcing us to ignore one term in the covariance matrix of estimated factorial effects while formulating the optimization problem. \cite{BA2018} proposed a model-based approach to overcome this problem in two-armed experiments, in which information on the correlation among the outcomes is available pre-intervention. Such an idea may be extended to the setting of factorial experiments.
	
	Also, a natural extension of the randomization-based framework of causal inference is the Bayesian framework, in which the potential outcomes are assumed to follow a hierarchical probabilistic model containing hyperparameters with assumed prior distributions. The Bayesian framework proposed in \cite{DPR2015} for drawing both super-population and finite-population causal inference from $2^K$ factorial designs can be utilized to obtain Bayesian optimal deigns according to different criteria proposed in literature \citep[e.g.,][]{ChalVerd1995}. 
	
	Another setting that has gained a lot of attention in recent times is when SUTVA is violated, for example, in the presence of interference between units. Extending the proposed results to such settings is a challenging, yet rewarding problem.
	
	Finally, in certain situations, it is possible that instead of the traditional factorial effects defined by (\ref{eq:popleveltau}), the experimenter is interested in other contrasts of the treatment means or more general factorial effects. One such natural choice of contrast is one that compares the  outcome of the control group with the average of all other groups that have at least one treatment.  Optimal allocations under such a reformulated optimization problem would be an interesting problem to study.
	
	
	\vspace{0.2 in}
	\noindent \textbf{Acknowledgement:} This research was partially supported by National Science Foundation grant SES 2217522.
	Any opinions, findings, conclusions, or recommendations expressed in this material are those of the authors and do not necessarily reflect the views of the National Science Foundation.
	
	\bibliographystyle{apalike}
	\bibliography{RPD_ref}

	\begin{appendices}
		
		\newpage
		\setcounter{page}{1}
		\begin{center}
			{\bf \Large  Supplementary Material\\for\\``Optimal allocation of sample size for randomization-based inference from $2^K$ factorial designs''}
		\end{center}

		\section{Proof of results}
		\label{App:ProofofResults}
		
		\noindent \textbf{Proof of Theorem~\ref{thm:optimal_allocation}}:
		
		\medskip
		
		\noindent We first state and prove a lemma about the matrix $\widetilde{\bm{V}}$.
		\begin{lemma} \label{lem:eigen_Vtilde}
			The matrix $\widetilde{\bm{V}}$ has $J$ non-zero eigenvalues $J (S^2_1/N_1), \ldots, J (S^2_J/N_J)$.
		\end{lemma}
		
		\noindent Proof: Note that $\widetilde{\bm{V}}$ can be written as  $\left( \sqrt{J}^{-1} \bm{L}^{\T} \right) \left(J \bm{A} \right) \left( \sqrt{J}^{-1} \bm{L} \right)$. Because the rows of $\sqrt(J)^{-1} \bm{L}$ form an orthonormal basis of vectors in $\mathbb{R}^J$, the above expression is a spectral decomposition of $\widetilde{\bm{V}}$. Thus, the eigenvalues of $\widetilde{\bm{V}}$ are its diagonal elements $J (S_1^2/N_1), \ldots, J (S_J^2/N_J)$. $\qed$
		
		\bigskip
		
		\noindent Now we prove the three parts of Theorem~\ref{thm:optimal_allocation}.
		
		\medskip
		
		\noindent Part (a): Because the trace of a matrix is the sum of its eigenvalues, from Lemma \ref{lem:eigen_Vtilde}, it follows that 
		$\tr(\widetilde{\bm{V}}) = J \sum_{j=1}^J (S_j^2/N_j)$. The problem of minimizing $\tr(\widetilde{\bm{V}})$ can be expressed as:
		$$ \text{minimize} \sum_{j=1}^J (S_j^2/N_j) \;\ \text{subject to} \sum_{j=1}^J N_j = N$$
		
		This constrained optimization problem can be obtained using the method of Lagrange multipliers, by solving
		\[ \min  \left[\sum_{j=1}^{J} S^2_j/N_j + \lambda \left( \sum_{j=1}^{J} N_j - N \right) \right],
		\]
		where $\lambda$ is a Lagrangian multiplier. Taking partial derivative of the objective function with respect to $N_j$ and setting it to zero, we get
		$ -S^2_j/N_j^2 + \lambda = 0$, which implies $N_j = S_j/\sqrt{\lambda}$.
		Solving for the constraint $\sum_{j=1}^{J} N_j = N$, we get,
		\[
		\sqrt{\lambda} = \frac{\sum_{j=1}^{J} S_j}{N} \Rightarrow N_j = S_j \left(\frac{N}{\sum_{j=1}^{J} S_j}\right) \Rightarrow p_j \propto S_j. \qed
		\]
		
		\medskip
		
		\noindent Part (b): Because the determinant of a matrix is the product of the eigenvalues, from Lemma \ref{lem:eigen_Vtilde}, we have $ |\widetilde{\bm{V}}| = J^{J} \prod_{j=1}^J (S_j^2/N_j). $
		The problem of minimizing the determinant can thus be equivalently expressed as
		$$ \text{minimize} \sum_{j=1}^J \log (S_j^2/N_j) \;\ \text{subject to} \;\ \sum_{j=1}^J N_j = N.$$
		Taking partial derivative of the objective function with respect to $N_j$ and setting it to zero, we get
		$ -(S^2_j/N_j^2)(N_j/S_j^2) + \lambda = -1/N_j + \lambda = 0$.
		It is straightforward to see that the optimal solution is $N_j = N/J$ or equivalently $p_j = 1/J$ for $j=1, \ldots, J$. $\qed$
		
		\medskip
		
		\noindent Part (c): Let $(N_1^{\mathcal{D}}, \ldots, N_J^{\mathcal{D}})$ denote the allocation vector of a design $\mathcal{D}$ . Also, let the $J$ eigenvalues of $\widetilde{\bm{V}}$ for design $\mathcal{D}$ be $\nu_1^{\mathcal{D}}, \ldots, \nu_J^{\mathcal{D}}$, and $\nu_{(1)}^{\mathcal{D}} \le \ldots \le \nu_{(J)}^{\mathcal{D}}$ denote the ordered eigenvalues. We will show that the design $\mathcal{D}^*$ in which $N_j^{\mathcal{D^*}} = (N  S^2_j)/ \left( \sum_{j=1}^J S^2_j \right)$ for $j=1, \ldots, J$ is the E-optimal design. From Lemma~\ref{lem:eigen_Vtilde}, design $\mathcal{D}^*$ can be characterized as a design in which all eigenvalues are equal to $J \left( \sum_{j=1}^J S^2_j \right) / N = \nu^{\mathcal{D}^*}$. Recall the definition of E-optimality of minimizing the maximum eigenvalue of the design matrix. Thus, it suffices to show that any design in which all eigenvalues of $\widetilde{\bm{V}}$ are not equal cannot be E-optimal because the solution $\mathcal{D}^*$ with all eigenvalues equal to $J \left( \sum_{j=1}^J S^2_j \right) / N $ is the only solution  with all eigenvalues equal in the feasible space of the optimization problem.
		
		We will proceed by proof by contradiction.
		Assume that a design $\Dtil$ for which all eigenvalues of $\widetilde{\bm{V}}$ are not equal is E-optimal, i.e.,
		\begin{equation}
			\nu_{(J)}^{\Dtil} \le \  \nu_{(J)}^{\mathcal{D}}, \label{eq:Eoptcond1} 
		\end{equation}
		for any design $\mathcal{D}$.
		
		Let $\mathcal{M}$ denote the set of $m \geq 1$ equal maximum eigenvalues of $\widetilde{\bm{V}}$ for design $\Dtil$ ($m=1$ indicate a unique maximum and because not all eigenvalues are equal we must have $m < J$). Then for any $j_1 \in \mathcal{M}$ and any $j_2 \notin \mathcal{M}$, $\nu_{j_1}^{\Dtil} > \nu_{j_2}^{\Dtil}$. Construct a new design $\Dprime$ by perturbing only the allocations for \emph{all} treatments $j_1 \in \mathcal{M}$ and \emph{one specific} $j_2 \notin \mathcal{M}$ as follows:
		\begin{eqnarray*}
			N_{j_1}^{\Dprime} &=& S^2_{j_1} \frac{\sum_{j \in \mathcal{M}} N_{j}^{\Dtil} + N_{j_2}^{\Dtil}}{ \sum_{j \in \mathcal{M}}S^2_{j} + S^2_{j_2}} \ \mbox{for all} \ j_1 \in \mathcal{M}  \\
			N_{j_2}^{\Dprime} &=& S^2_{j_2}  \frac{\sum_{j \in \mathcal{M}} N_{j}^{\Dtil} + N_{j_2}^{\Dtil}}{ \sum_{j \in \mathcal{M}}S^2_{j} + S^2_{j_2}}\\
			N_j^{\Dprime} &=& N_j^{\Dtil}, \ \mbox{for} \ j \notin \mathcal{M} \cup \{j_2\}. \qed
		\end{eqnarray*}
		
		Using Lemma \ref{lem:eigen_Vtilde}, after a little algebra, it follows that $ \nu_{j_1}^{\Dtil} > \nu_{j_1}^{\Dprime} = \nu_{j_2}^{\Dprime} > \nu_{j2}^{\Dtil}$ for all $j_1 \in \mathcal{M}$. Also, for all $j \notin \mathcal{M} \cup \{j_2\}$, $\nu_{j_1}^{\Dtil} > \nu_{j}^{\Dtil} = \nu_{j}^{\Dprime}$. Consequently $\nu_{(J)}^{\Dtil} > \  \nu_{(J)}^{\Dprime}$, and contradicts (\ref{eq:Eoptcond1}).
		
		\medskip
		
		\hrule
		
		\bigskip
		
		\noindent \textbf{Proof of Theorem~\ref{thm:optimal_allocation_blk}}:
		
		\medskip
		We need the following lemma regarding the eigenvalues of the covariance matrix $\widetilde{\bm{V}}_{\text{blk}}$ defined in (\ref{eq:blk_matrix_criterion}) under a blocked design:
		
		\begin{lemma} \label{lem:eigen_Vtilde_blk}
			The matrix $\widetilde{\bm{V}}_{\text{blk}}$ has $J$ non-zero eigenvalues $J S^2_{\text{blk}, 1}, \ldots, J S^2_{\text{blk}, J}$.
		\end{lemma}
		
		\noindent Proof: Along the same lines as Lemma~\ref{lem:eigen_Vtilde}, note that $\widetilde{\bm{V}}_{\text{blk}}$ can written as  $\left( \sqrt{J}^{-1} \bm{L}^{\T} \right) \left(J \bm{A}_{\text{blk}} \right) \left( \sqrt{J}^{-1} \bm{L} \right)$. Because the rows of $\sqrt(J)^{-1} \bm{L}$ forms an orthonormal basis of vectors in $\mathbb{R}^J$, the above expression is a spectral decomposition of $\widetilde{\bm{V}}_{\text{blk}}$. Thus, the eigenvalues of $\widetilde{\bm{V}}_{\text{blk}}$ are its diagonal elements $J S^2_{\text{blk}, 1}, \ldots, J S^2_{\text{blk}, J}$. $\qed$
		
		\bigskip
		
		\noindent Now we prove the three parts of the main theorem.
		
		\noindent Part (a): We can proceed very similarly to the proof of Theorem~\ref{thm:optimal_allocation}(a). Because the trace of a matrix is the sum of its eigenvalues, using Lemma~\ref{lem:eigen_Vtilde_blk},
		$\tr(\widetilde{\bm{V}}_{\text{blk}}) = J \sum_{j=1}^J  S^2_{\text{blk}, j}$. The problem of minimizing $\tr(\widetilde{\bm{V}}_{\text{blk}})$ can be expressed as
		$$ \text{minimize} \sum_{j=1}^J  S^2_{\text{blk}, j} \;\ \text{subject to} \sum_{j=1}^J M_{h,j} = M_h \ \text{for } h=1,\dots,H.$$
		
		We can again solve this using the method of Lagrange multipliers,
		\begin{align*}
			\min  \left[\sum_{j=1}^{J} S^2_{\text{blk}, j} + \sum_{h=1}^H\lambda_h \left( \sum_{j=1}^{J} M_{h,j} - M_h \right) \right]=\min  \left[\sum_{j=1}^{J} \sum_{h=1}^H \frac{M_h^2}{N^2}\frac{S^2_{h,j}}{M_{h,j}} + \sum_{h=1}^H\lambda_h \left( \sum_{j=1}^{J} M_{h,j} - M_h \right) \right],
		\end{align*} 
		where $\lambda_h$ are Lagrangian multipliers.
		Taking partial derivative of the objective function with respect to $M_{h,j}$ and setting it to zero, we get
		$ -\frac{M_h^2}{N^2}\frac{S^2_{h,j}}{M_{h,j}^2} + \lambda_h = 0$, which implies 
		\[M_{h,j} = \frac{M_h}{N}\frac{S_{h,j}}{\sqrt{\lambda_h}}.\]
		Solving for the constraint $\sum_{j=1}^{J} M_{h,j} = M_h$, we get,
		\[
		\sqrt{\lambda_h} = \frac{\sum_{j=1}^{J} S_{h,j}}{N} \Rightarrow M_{h,j} =  M_h\frac{S_{h,j}}{\sum_{j=1}^{J} S_{h,j}} \Rightarrow p_{h,j} = \frac{S_{h,j}}{\sum_{j=1}^{J} S_{h,j}}. \qed
		\] 
		
		\bigskip
		
		\noindent Part (b): Because the determinant of a matrix is the product of the eigenvalues, from Lemma \ref{lem:eigen_Vtilde_blk}, we have
		$$ |\widetilde{\bm{V}}| = J^{J} \prod_{j=1}^J S^2_{\text{blk}, j} \quad . $$
		The problem of minimizing the determinant can thus be equivalently expressed as
		$$ \text{minimize} \sum_{j=1}^J \log (S^2_{\text{blk}, j}) \;\ \text{subject to} \;\ \sum_{j=1}^J M_{h,j} = M_h \quad .$$
		
		To prove the special cases, we use the Lagrangian multiplier based optimization approach as in (a). So, we need to solve
		\begin{align*}
			&\min  \left[\sum_{j=1}^{J} \log(S^2_{\text{blk}, j}) + \sum_{h=1}^H\lambda_h \left( \sum_{j=1}^{J} M_{h,j} - M_h \right) \right]\\
			&=\min  \left[\sum_{j=1}^{J}\log\left( \sum_{h=1}^H \frac{M_h^2}{N^2}\frac{S^2_{h,j}}{M_{h,j}}\right) + \sum_{h=1}^H\lambda_h \left( \sum_{j=1}^{J} M_{h,j} - M_h \right) \right],
		\end{align*} 
		where $\lambda_h$ are Lagrangian multipliers.
		
		Taking the derivative with respect to $M_{h,j}$ and setting equal to 0, we have
		\begin{align}
			0=   -M_h^2\frac{S_{h,j}^2}{M_{h,j}^2}\left(\sum_{k=1}^HM_k^2\frac{S_{k,j}^2}{M_{k,j}}\right)^{-1} + \lambda_h \quad .
			\label{pf:lagrange_cond} 
		\end{align}

		\noindent \textbf{Under special case (i), WBH: variances of potential outcomes are the same within each block (such that $S^2_{h,j} = S^2_{h,\cdot}$ for all $j=1,\dots,J$).}
		
		\noindent Substitution of $S^2_{h,j} = S^2_{h,\cdot}$ into (\ref{pf:lagrange_cond}) yields
		$$0 = -M_h^2\frac{S_{h,\cdot}^2} {M_{h,j}^2}\left(\sum_{k=1}^HM_k^2\frac{S_{k,j}^2}{M_{k,j}}\right)^{-1} + \lambda_h.$$
		
		After a little algebra, we get
		\begin{equation}
			M_{h,j} =  M_h\frac{S_{h,\cdot}}{\sqrt{\lambda_h}} c_j \quad \text{ where } c_j =\sqrt{\left(\sum_{k=1}^HM_k^2\frac{S_{k,\cdot}^2}{M_{k,j}}\right)^{-1} } \label{eq:eq_RBD_Di_1} 
		\end{equation}
		
		Now, summing over $j$ in (\ref{eq:eq_RBD_Di_1}) and applying the second constraint $\sum_{j=1}^J M_{h,j} = M_h$ we have:
		\begin{eqnarray}
			&& M_h = \sum_{j=1}^J M_{h,j} =  \sum_{j=1}^JM_h\frac{S_{h,\cdot}}{\sqrt{\lambda_h}}c_j =  \frac{S_{h,\cdot}M_h}{\sqrt{\lambda_h}}\sum_{j=1}^Jc_j \nonumber \\
			&\Rightarrow& \frac{S_{h,\cdot}}{\sqrt{\lambda_h}} =   \left(\sum_{j=1}^Jc_j\right)^{-1} \label{eq:eq_RBD_Di_2} 
		\end{eqnarray}
		Thus, substituting (\ref{eq:eq_RBD_Di_2}) in (\ref{eq:eq_RBD_Di_1}),
		\begin{align*}
			M_{h,j} = M_h \frac{c_j}{\sum_{j=1}^J c_j} = M_h \delta_j \text{ where } \delta_j =\frac{c_j}{\sum_{j=1}^J c_j} \ .  
		\end{align*}
		
		Substituting this back into definition of $c_j$ in (\ref{eq:eq_RBD_Di_1}) gives us,
		\begin{align*}
			c_j &=  \sqrt{\left(\sum_{k=1}^HM_k^2\frac{S_{k,\cdot}^2}{M_k \delta_j}\right)^{-1}} 
			= \sqrt{\left(\sum_{k=1}^HM_k\frac{S_{k,\cdot}^2}{\delta_j}\right)^{-1}} 
			= \sqrt{\delta_j} \sqrt{\left(\sum_{k=1}^H M_k S_{k,\cdot}^2\right)^{-1}} \\ 
			&= \sqrt{\delta_j} \alpha \quad \text{where} \quad \alpha = \sqrt{\left(\sum_{k=1}^H M_k S_{k,\cdot}^2\right)^{-1}} \quad \text{is a constant}\\
			&= \frac{\sqrt{c_j}}{\sqrt{\sum_{j=1}^J{c_j}}}\alpha \\
			&\implies c_j = \frac{\alpha^2}{\sum_{j=1}^J{c_j}} = \beta, \quad \text{a constant free of $j$}.
		\end{align*}
		
		Thus, $\delta_j = c_j / \sum_j c_j = 1/J$ and
		$M_{h,j} = M_h/J$ and $p_{h,j}=1/J$. $\qed$

		\noindent \textbf{Under special case (ii), BBH: variances are the same across blocks for each treatment (such that $S^2_{h,j} = S^2_{\cdot,j}$ for all $h=1,\dots,H$).}
		
		Substituting $S^2_{h,j} = S^2_{\cdot,j}$ in (\ref{pf:lagrange_cond}), we get, 
		\begin{align}
			&0=   -\frac{M_h^2}{M_{h,j}^2}\left(\sum_{k=1}^H\frac{M_k^2}{M_{k,j}}\right)^{-1} + \lambda_h \nonumber\\
			\implies &M_{h,j} =  \frac{M_h}{\sqrt{\lambda_h}}\sqrt{\left(\sum_{k=1}^H\frac{M_k^2}{M_{k,j}}\right)^{-1} } =  \frac{M_h}{\sqrt{\lambda_h}} c_j \quad \text{ where } c_j =\sqrt{\left(\sum_{k=1}^H\frac{M_k^2}{M_{k,j}}\right)^{-1} }  \label{eq:eq_RBD_Dii_1}
		\end{align}
		
		Now, summing over $j$ and applying the second constraint $\sum_{j=1}^J M_{h,j}$ yields:
		\begin{align*}
			&M_h = \sum_{j=1}^J M_{h,j} =  \sum_{j=1}^J\frac{M_h}{\sqrt{\lambda_h}}c_j = \frac{M_h}{\sqrt{\lambda_h}} \sum_{j=1}^J c_j \\
			\implies &\frac{M_h}{\sqrt{\lambda_h}} = \frac{M_h}{\sum_{j=1}^Jc_j}.
		\end{align*}
		
		Proceeding similarly as in the proof of the previous part, we can show that $\delta_j = c_j/(\sum_{j=1}^Jc_j) = 1/J$ for all $j = 1, \ldots, J$ and hence $M_{h,j} = M_h/J$ and $p_{h,j}=1/J$. $\qed$
		
		\bigskip
		
		\noindent Part (c):
		Recall from Section~\ref{sec:RBDfactorial}, $
		S^2_{blk,j} = \sum_{h=1}^H (M_h/N)^2(S^2_{h,j}/M_{h,j})$.
		
		\noindent We need the following lemmas to prove the theorem.
		\begin{lemma}\label{lem:Eoptalg1}
			For any $r>0$ and $J>1$, $\delta_j> -r$ such that $\sum_{j=1}^J \delta_j = 0$ and $\delta_j \ne 0 \ \forall j=1,2,\dots,J$,
			\[
			\sum_j \biggl(\frac{1}{r+\delta_{j}} - \frac{1}{r} \biggr) > 0
			\]
		\end{lemma}
		
		\begin{proof} Let $j^* = \argmin_{\delta_j>0} \delta_j$. Then,
			
			\begin{itemize}
				\item when $\delta_j > 0 $ and $j \ne j^*$, $
				(r + \delta_j^*) < (r + \delta_j) \implies 1/(r + \delta_j^*) > 1/(r + \delta_j) \implies -\delta_j/(r + \delta_j^*) < -\delta_j/(r + \delta_j)$
				\item when $\delta_j < 0 $, $
				(r + \delta_j^*) > (r + \delta_j) \implies 1/(r + \delta_j^*) < 1/(r + \delta_j) \implies -\delta_j/(r + \delta_j^*) < -\delta_j/(r + \delta_j)$
				\item when $\delta_j = 0 $, $
				-\delta_j/(r + \delta_j^*) = -\delta_j/(r + \delta_j)$
				\item when $j = j^*$, $
				-\delta_j/(r + \delta_j^*) = -\delta_j/(r + \delta_j)$ 
			\end{itemize}
			Thus, it holds that,
			\begin{align*}
				\sum_j \frac{-\delta_j}{r + \delta_j} > \sum_j \frac{-\delta_j}{r + \delta_{j^*}}
			\end{align*}
			
			Then,
			\begin{align*}
				\sum_j \biggl(\frac{1}{r+\delta_{j}} - \frac{1}{r} \biggr) = \sum_j \frac{-\delta_{j}}{r(r+\delta_{j})} > \sum_j \frac{-\delta_{j}}{r(r+\delta_{j^*})} = \frac{-\sum_j \delta_{j}}{r(r+\max_j\delta_{j^*})} = 0
			\end{align*}
		\end{proof}
		
		\begin{lemma}\label{lem:min_blockconstraints}
			Given integers $M_h$ for $h=1,\dots,H$, let $M_{h,j}$ denote any allocation of $M_h$ into $J>1$ groups such that $\sum_j M_{h,j} = M_h$. Let $a_h>0$, for $h \in \{1,\dots,H\}$, be fixed. Then,  
			\[\max_j \sum_{h=1}^H \frac{a_h}{M_{h,j}} \ge \sum_h \frac{a_h}{(\frac{M_h}{J})}.\] 
			That is, the allocation that minimizes $\max_j \sum_{h=1}^H a_h/M_{h,j}$ is $M_{h,j} = M_h/J$.
		\end{lemma}
		\begin{proof}
			
			For each $h=1,\dots,H$, using an argument similar to the one in the proof of Theorem~\ref{thm:optimal_allocation}, we can write for $h=1,\dots,H$,
			\[
			\max_j \frac{a_h}{M_{h,j}} \ge \frac{a_h}{\frac{M_h}{J}}
			\]
			
			Suppose there exists an allocation $\tilde M_{h,j} = M_h/J + \delta^h_j$ that maximizes $\max_j \sum_{h=1}^H a_h/M_{h,j}$, where for each $h$, $\exists j$ such that $ \delta^h_j \ne 0$. Then, $\sum_{j=1}^J \tilde{M}_{h,j}= M_h \implies \sum_{j=1}^J \delta^h_j = 0 \ \forall h = 1,\dots,H$. Further, by our assumption on $\tilde M_{h,j}$,
			\begin{align*}
				\max_j \sum_{h=1}^H \frac{a_h}{\tilde M_{h,j}} \le \sum_{h=1}^H \frac{a_h}{\frac{M_h}{J}} &\iff \max_j \sum_{h=1}^H \biggl(\frac{a_h}{\tilde M_{h,j}} - \frac{a_h}{\frac{M_h}{J}}\biggr) \le 0 \iff \max_j \sum_{h=1}^H \biggl(\frac{a_h}{(\frac{M_h}{J}+\delta^h_j)} -  \frac{a_h}{\frac{M_h}{J}}\biggr) \le 0  \\ 
				&\iff \max_j \sum_{h=1}^H \epsilon^h_j \le 0 \text{ (where } \epsilon^h_{j} =  \frac{a_h}{(\frac{M_h}{J}+\delta^h_j)} - \frac{a_h}{\frac{M_h}{J}})\\
				&\iff \sum_{h=1}^H \epsilon^h_j \le 0 \quad \forall j \\
				&\implies \sum_{j=1}^J \sum_{h=1}^H \epsilon^h_j \le 0 \\
				&\iff \sum_{h=1}^H a_h \sum_{j=1}^J \biggl(\frac{1}{(\frac{M_h}{J}+\delta^h_j)} - \frac{1}{\frac{M_h}{J}}\biggr) \le 0 
			\end{align*}
			which  is a contradiction by taking $r= M_h/J$ and  $\delta_j = \delta^h_j$ for each $h$ in Lemma~\ref{lem:Eoptalg1}.
			
			\noindent Thus, we have that,
			\[
			\max_j \sum_h \frac{a_h}{M_{h,j}} \ge \sum_h \frac{a_h}{\frac{M_h}{J}}
			\]
		\end{proof}

		\noindent \textbf{Under special case, WBH: variances of potential outcomes are the same within each block (such that $S^2_{h,j} = S^2_{h,\cdot}$ for all $j=1,\dots,J$).}
		
		\noindent We have,
		$S^2_{blk,j} = \sum_{h=1}^H (M_h/N)^2(S^2_{h,.}/M_{h,j})$. 
		We can rewrite this equation as $S^2_{blk,j}= \sum_{h=1}^H a_h/M_{h,j}$ where $a_h = (M_h/N)^2S^2_{h,.}$ does not depend on $j$.

		\noindent By Lemma~\ref{lem:min_blockconstraints}, we get $\max_j \sum_{h=1}^H a_h/M_{h,j} \ge J \sum_h a_h/M_h$. 
		
		We immediately see that for $M^*_{h,j} = M_h/J$, the left hand side of the inequality attains the lower bound, which is minimax. Hence, $p^*_{h,j} = 1/J$ for each treatment $j$ within block $h$, a balanced design within each block.  \qed
		
		\medskip
		
		\section{Conditions for Greedy Algorithm from \cite{FRIEDRICH20151}} \label{conditionsGreedy}
		Theorem~3.2 of \cite{FRIEDRICH20151} have the following conditions to be satisfied for global convergence of the greedy algorithm to the true optimum. We restate the theorem here for reference.
		
		\begin{theorem3.2}
			
			The globally optimal integer solution of the problem 
			\[ \min_{(N_1,\dots,N_J)} \biggl\{f(N_1,\dots,N_J) \ \bigg| \ N_j\ge0 \ \forall j, \ \sum_{j\in A} N_j \le \varphi(A) \ \forall A\subset E \biggr\},
			\]
			is found by a Greedy algorithm if
			\begin{enumerate}
				\item E is a finite set,
				\item $\varphi : 2^E \rightarrow \mathbb{Z}_+$ is submodular, monotone and satisfies $\varphi(\phi) = 0$,
				\item $f : \mathbb{R}^E_+ \rightarrow \mathbb{R}$ is separable and convex with continuous components.
			\end{enumerate}
		\end{theorem3.2}
		
		\medskip
		
		\noindent In the case of a CRD, we can identify the following in the theorem above:
		\begin{itemize}
			\item $E=\{1,2,..J\}$ (finite)
			\item $A \in \{\{\emptyset\},\{1\},...,\{J\},\{1,2\},..,\{1,...,J\}\} = 2^E$
			\item $\varphi(A) = \min(\sum_{j\in A}(u_j-l_j), N - \sum_{j\in E}l_j)$
			\item $f(N_1,\dots,N_J)$ defined as in Section~\ref{sec:CRDfactorial}, $\sum_j S_j^2/N_j$ for A-optimality and $\sum_j \log(S_j^2/N_j)$ for D-optimality respectively (separable and convex) and each continuous in their individual components, i.e., $f_j(N_j)$ are continuous in $N_j$ ($1/x$ and $\log(1/x)$ is continuous in $x$ for $x \in \mathbb{R}^E_+$).
		\end{itemize} 
		
		\noindent Condition~1 is already satisfied as noted above. Condition~3 is straightforward since all our real-valued objective functions are convex and finite-dimensional and hence separable. Thus, it suffices to show that Condition~2 above in Theorem~3.2 is satisfied in our case for the submodular set function $\varphi(.)$. Again, we give the definition of a submodular function as in \cite{FRIEDRICH20151}.

		\begin{Definition}[Submodular function]
			$\varphi : 2^E \rightarrow \mathcal{Z}_+$ is submodular if 
			\[
			\varphi(X\cap Y) + \varphi(X\cup Y) \le \varphi(X) + \varphi(Y), \ \forall X,Y\subset E
			\]
		\end{Definition}
		\vskip 1em
		
		Thus, in the case of a CRD, take $A$ to be defined as above, then $g(A) = \sum_{j \in A} (u_j-l_j)$, $h(A) = N - \sum_{j \in A} l_j$ as corresponding set functions for the following argument. Linear set functions are always submodular as can be quickly shown: $g(X\cup Y) = \sum_{j\in X}(u_j-l_j) + \sum_{j\in Y}(u_j-l_j) - \sum_{j\in X\cap Y} (u_j-l_j) = g(X) + g(Y) - g(X \cap Y)$. Similarly for $h$, $h(X\cup Y) = (N-\sum_{j\in X} u_j) + (N-\sum_{j\in Y}u_j) - (N-\sum_{j\in X\cap Y} u_j) = h(X) + h(Y) - h(X \cap Y)$. Submodularity of our $\phi(.)$ function follows directly from \cite{FRIEDRICH20151} as $\min(g,h)$ is submodular if $g,h$ are submodular and $g-h$ is monotone. $g$ is submodular and monotone (defined as $\forall T,S \subset E, \text{ s.t. } T \subset S \Rightarrow f(T) \le f(S)$) because $u_j-l_j>0, \ \forall j\in J$ . And, $h$ is submodular. Finally, $(g-h)(A) = \sum_{j\in A}u_j - N$ is monotone since $g-h$ is linear in $A$. 
		
		\section{Empirical evidence of greedy algorithm for D-optimality in RBD} \label{App: EmpEvidenceGreedyBlk} 
		Following from Section~\ref{sec:RBDfactorial},  
		we show the performance of the greedy algorithms for finding D-optimal solutions empirically.
		
		Five different settings of $2^2$ factorial designs in two blocks, each corresponding to a specific type of potential outcome matrix, are considered. Each setting is defined by the block sizes $M_1$ and $M_2$, and the $4 \times 2$ matrix of variances $S^2_{h,j}$ as shown in columns 3 and 4 of Table~\ref{tbl:D-opt_GreedySim}, respectively.
		
		The first setting considers blocks of equal sizes with potential outcomes satisfying Condition~\ref{cond_str_add} (strict additivity), leading to a D-optimal design that is balanced within each block as per Remark~\ref{remark:str_add_blk1}. The second setting considers equal block sizes with potential outcomes satisfying Condition~\ref{cond_WBH} (WBH). The third setting considers unequal block sizes with potential outcomes satisfying Condition~\ref{cond_BBH} (BBH) but not Condition~\ref{cond_WBH} (WBH). Note that, in the above cases, the exact solution as given by Theorem~\ref{thm:optimal_allocation_blk} is indeed an integer solution, due to the choice of $M_h$ and $J$. The fourth setting is similar to the third case above, but is one where Theorem~\ref{thm:optimal_allocation_blk} provides an exact solution that is not an integer solution. An exhaustive search leads to identification of six different allocations, each of which is optimal. In this case, the greedy algorithm identifies one of these solutions. The fifth setting satisfies neither Condition~\ref{cond_WBH} (WBH) nor Condition~\ref{cond_BBH} (BBH), and consequently Theorem~\ref{thm:optimal_allocation_blk} cannot provide an exact D-optimal solution. However, the greedy algorithm identifies the true optimal integer allocation (where truth is identified through exhaustive search).
		
		\begin{table}[ht]
			\centering \footnotesize
			\begin{tabular}{c|m{4cm}|m{1cm}|m{2.5cm}|m{3.5cm}|m{3.5cm}}
				S.No. & Case & Block Size ($M_h$) & Variances ($S_{h,j}^2$)  & Exhaustive search \newline E-optimal solution & Greedy solution \\[3ex]
				\hline & & & & & \\
				1. & Equal blocks with equal variances & 
				$\begin{bmatrix} 40 \\40 \end{bmatrix}$ & 
				$\begin{bmatrix} 1 & 1 & 1 & 1 \\ 1 & 1 & 1 & 1 \end{bmatrix}$ &  $\begin{bmatrix} 10 & 10 & 10 & 10 \\ 10& 10 & 10 & 10 \end{bmatrix} $ & $\begin{bmatrix} 10 & 10 & 10 & 10 \\ 10& 10 & 10 & 10 \end{bmatrix}$ \\ 
				& & & & & \\
				2. & Equal blocks with equal variances for all treatments within block & 
				$\begin{bmatrix} 40 \\40 \end{bmatrix}$ & 
				$\begin{bmatrix} 4 & 4 & 4 & 4 \\ 1 & 1 & 1 & 1 \end{bmatrix}$ &  $\begin{bmatrix} 10 & 10 & 10 & 10 \\ 10& 10 & 10 & 10 \end{bmatrix} $ & $\begin{bmatrix} 10 & 10 & 10 & 10 \\ 10& 10 & 10 & 10 \end{bmatrix}$ \\
				& & & & & \\
				3. & Unequal blocks with equal variances across blocks for each treatment & 
				$\begin{bmatrix} 40 \\20 \end{bmatrix}$ & 
				$\begin{bmatrix} 1 & 2 & 3 & 4 \\ 1 & 2 & 3 & 4 \end{bmatrix}$ &  $\begin{bmatrix} 10 & 10 & 10 & 10 \\ 5 & 5 & 5 & 5 \end{bmatrix} $ & $\begin{bmatrix} 10 & 10 & 10 & 10 \\ 5 & 5 & 5 & 5 \end{bmatrix}$ \\
				& & & & & \\
				4. & Unequal blocks with equal variances but exact solution is non-integer & 
				$\begin{bmatrix} 40 \\20 \end{bmatrix}$ & 
				$\begin{bmatrix} 1 & 2 & 3 & 5 \\ 1 & 2 & 3 & 5 \end{bmatrix}$ &  $\begin{bmatrix} 
					\begin{bmatrix} 10 & 10 & 10 & 10\\8& 8 & 7 & 7 \end{bmatrix} \\[1em]
					\begin{bmatrix} 10 & 10 & 10 & 10\\8& 7 & 8 & 7 \end{bmatrix}\\[1em]
					\begin{bmatrix} 10 & 10 & 10 & 10\\7& 8 & 8 & 7 \end{bmatrix}\\[1em]
					\begin{bmatrix} 10 & 10 & 10 & 10\\8& 7 & 7 & 8\end{bmatrix}\\[1em]
					\begin{bmatrix} 10 & 10 & 10 & 10\\7& 8 & 7 & 8\end{bmatrix}\\[1em]
					\begin{bmatrix} 10 & 10 & 10 & 10\\7& 7 & 8 & 8\end{bmatrix}
				\end{bmatrix}$ & 
				$\begin{bmatrix}10 & 10 & 10 & 10\\8& 7 & 8 & 7 \end{bmatrix}$ \\
				& & & & & \\
				5. & Equal blocks with unequal variances & 
				$\begin{bmatrix} 40 \\20 \end{bmatrix}$ & 
				$\begin{bmatrix} 1 & 2 & 3 & 4 \\ 4& 3 & 2& 1 \end{bmatrix}$ &  $
				\begin{bmatrix} 7 & 10 & 11 & 12 \\ 7 & 6 & 4 & 3 \end{bmatrix} $ & $\begin{bmatrix} 7 & 10 & 11 & 12 \\ 7 & 6 & 4 & 3  \end{bmatrix}$ 
			\end{tabular}
			\caption{Summary of Greedy algorithm solutions for D-optimality for $H=2, K=2$}
			\label{tbl:D-opt_GreedySim}
		\end{table}

		\section{Greedy Algorithm for RBD}\label{append:greed_alg}
		Using the methods in Section \ref{ss:computation}, we can use the same algorithm in Figure~\ref{fig:greedyalg} to obtain the optimal integer solutions for the A-optimality case under block design by taking $f_{h,j}(M_{h,j})= (M_h^2/N^2)(S_{h,j}^2/M_{h,j})$. The exact algorithm taking the structure of the blocks into account is given in Figure~\ref{fig:greedyalgBlk}.
		
		For D- and E-optimality, we extend the greedy idea from the previous parts and provide the algorithm in Figure~\ref{fig:greedyalgBlk} and \ref{fig:greedyalgEblk}. For D-optimality, take $f_{h,j}(M_{h,j})= \log(\sum_{h=1}^H(M_h^2/N^2)(S_{h,j}^2/M_{h,j}))$. For E-optimality, take $f_{h,j}(M_{h,j})= (M_h^2/N^2)(S_{h,j}^2/M_{h,j})$.
		
		\begin{figure}[htbp]
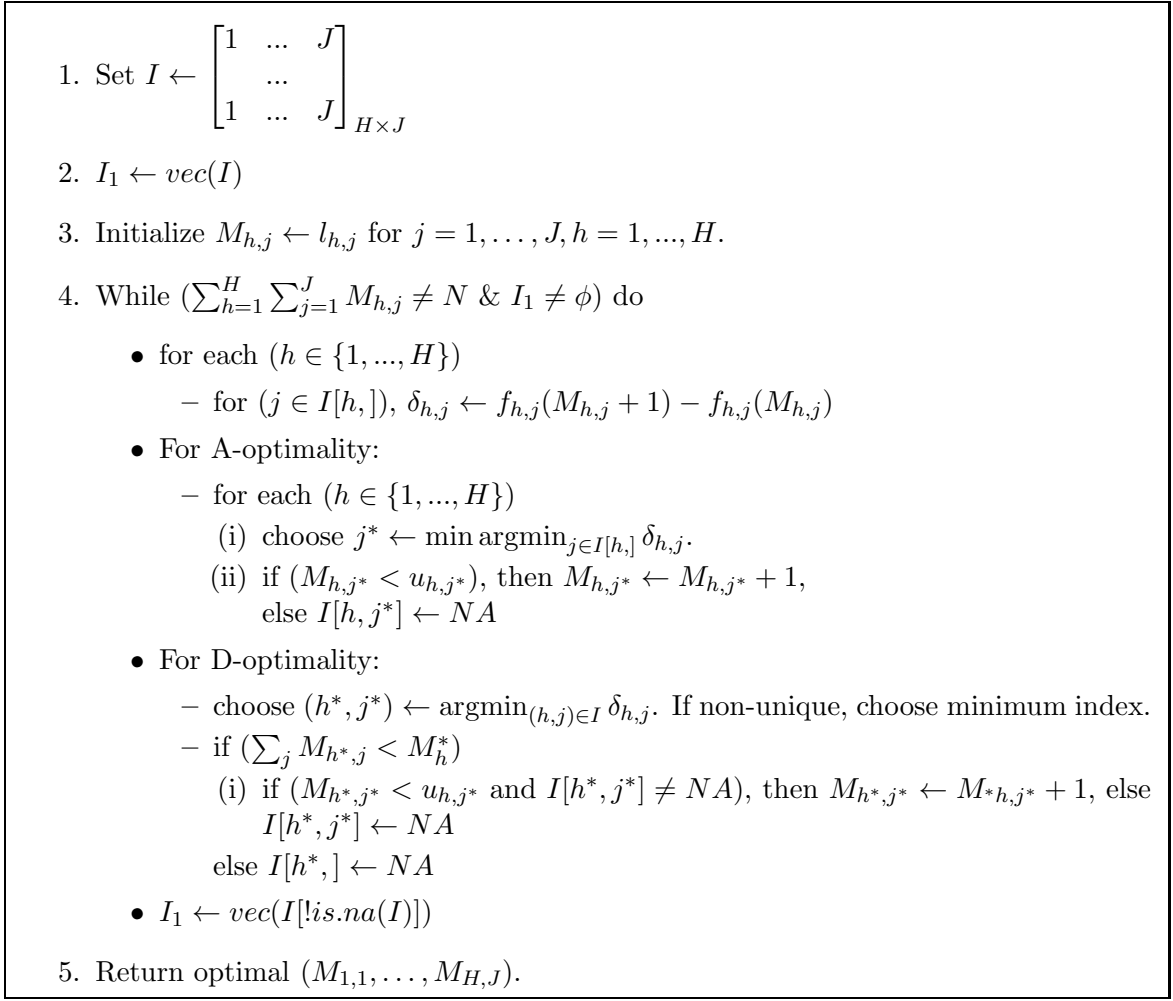

			\caption{Greedy algorithm for A- \& D-optimality} \label{fig:greedyalgBlk}
			\centering
			\fbox{
				\begin{minipage}{15 cm}
					\begin{enumerate}
						\item Set $I \leftarrow \begin{bmatrix}
							1 & ...& J \\
							& ... & \\
							1 & ...& J 
						\end{bmatrix}_{H\times J}$
						\item $I_1 \leftarrow vec(I)$
						\item Initialize $M_{h,j} \leftarrow l_{h,j}$ for $j=1, \ldots, J, h=1,...,H$. 
						\item While ($\sum_{h=1}^H\sum_{j=1}^J M_{h,j} \ne N \ \& \ I_1 \ne \phi$) do
						\begin{itemize}
							\item for each $(h \in \{1,...,H\})$
							\begin{itemize}
								\item for $(j \in I[h,])$, $\delta_{h,j} \leftarrow f_{h,j}(M_{h,j}+1) - f_{h,j}(M_{h,j})$
							\end{itemize}
							\item For A-optimality:
							\begin{itemize}
								\item for each $(h \in \{1,...,H\})$
								\begin{itemize}
									\item[(i)]  choose $j^* \leftarrow \min \argmin_{j \in I[h,]} \delta_{h,j}$.
									\item[(ii)]  if $(M_{h,j^*} < u_{h,j^*})$, then $M_{h,j^*}  \leftarrow M_{h,j^*}  + 1$, \\
									else $I[h,j^*] \leftarrow NA$
								\end{itemize}
							\end{itemize}
							
							\item For D-optimality:
							\begin{itemize}
								\item choose $(h^*,j^*) \leftarrow \argmin_{(h,j) \in I} \delta_{h,j}$. If non-unique, choose minimum index.
								\item if $(\sum_j M_{h^*,j} < M_{h}^*)$
								\begin{itemize}
									\item[(i)] if $(M_{h^*,j^*} < u_{h,j^*}$ and $I[h^*,j^*] \ne NA)$, then $M_{h^*,j^*}  \leftarrow M_{^*h,j^*}  + 1$, 
									else $I[h^*,j^*] \leftarrow NA$
								\end{itemize} 
								else $I[h^*,] \leftarrow NA$
							\end{itemize}
							
							\item $I_1 \leftarrow vec(I[!is.na(I)])$
						\end{itemize}
						\item Return optimal $(M_{1,1}, \ldots, M_{H,J})$.
					\end{enumerate}
				\end{minipage}
			}
		\end{figure}
		
		\begin{figure}[htbp]
			\caption{Greedy algorithm for E-optimality in the Block Case} \label{fig:greedyalgEblk}
			\centering
			\fbox{
				\begin{minipage}{15 cm}
					\begin{enumerate}
						\item Set $I \leftarrow \begin{bmatrix}
							1 & ...& J \\
							& ... & \\
							1 & ...& J 
						\end{bmatrix}_{H\times J}$
						\item $I_1 \leftarrow vec(I)$
						\item Initialize $M_{h,j} \leftarrow l_{h,j}$ for $j=1, \ldots, J, h=1,...,H$. 
						\item While ($\sum_{h=1}^H\sum_{j=1}^J M_{h,j} \ne N \ \& \ I_1 \ne \phi$) do
						\begin{itemize}
							\item for each $(h \in \{1,...,H\})$
							\begin{itemize}
								\item for $(j \in I[h,])$, $\delta_{h,j} \leftarrow f_{h,j}(M_{h,j}+1) - f_{h,j}(M_{h,j})$
							\end{itemize}
							\item choose $j^* \leftarrow \min \argmin_{j \in \{1,...,J\}} \sum_{h=1}^H f_{h,j}$
							\item choose $h^* \leftarrow \min \argmin_{h\in \{1,...,H\}} \delta_{h,j^*}$
							\item if $(\sum_j M_{h^*,j} < M_{h}^*)$
							\begin{itemize}
								\item if $(M_{h^*,j^*} < u_{h,j^*}$ and $I[h^*,j^*] \ne NA)$, then $M_{h^*,j^*}  \leftarrow M_{^*h,j^*}  + 1$, 
								else $I[h^*,j^*] \leftarrow NA$
							\end{itemize} 
							else $I[h^*,] \leftarrow NA$
							\item $I_1 \leftarrow vec(I[!is.na(I)])$
						\end{itemize}
						\item Return optimal $(M_{1,1}, \ldots, M_{H,J})$.
					\end{enumerate}
				\end{minipage}
			}
		\end{figure}

		\pagebreak
	\end{appendices}

\end{document}